\documentclass[10pt]{article}

\usepackage{graphicx} 
\usepackage[english]{babel}
\usepackage{amsmath}
\usepackage{amssymb}
\usepackage{amsthm}
\usepackage{subfig}
\usepackage{xcolor}
\usepackage{comment}
\usepackage{natbib}
\usepackage{hyperref}

\newtheorem{theorem}{Theorem}

\def\commenton{1}
\newcommand{\houssam}[1]{\if\commenton1{\color{green}(HN: #1)}\fi}
\newcommand{\abbas}[1]{\if\commenton1{\color{magenta}(AZ: #1)}\fi}
\newcommand{\richard}[1]{\if\commenton1{\color{teal}(RM: #1)}\fi}

\title{Breaking the Winner's Curse with Bayesian Hybrid Shrinkage}
\author{Richard Mudd, Abbas Zaidi, Rina Friedberg,\\Ilya Gorbachev, Anchal Choubey, and Houssam Nassif\\[6pt] \normalsize Meta Platforms}

\begin{document}

\maketitle

\begin{abstract}
    \noindent The widespread adoption of randomized controlled trials (A/B Tests) for decision-making has introduced a pervasive "Winner's Curse": experiments selected for launch often exhibit upwardly biased effect estimates and invalid confidence intervals. This selection bias leads to over-optimistic impact projections and undermines decision-making, particularly in low-power regimes. We propose Bayesian Hybrid Shrinkage (BHS), an empirical Bayes (EB) framework that leverages data-driven priors to mitigate selection bias and provides accurate uncertainty quantification. Unlike traditional EB methods that apply uniform shrinkage, BHS introduces an experiment-specific "local" shrinkage factor that incorporates individual experiment characteristics, improving robustness against prior misspecification. We also derive a closed-form inference strategy designed for high-throughput production environments. Extensive simulations and real-world evaluations at Meta Platforms demonstrate that BHS outperforms existing methods in terms of bias reduction and interval coverage, even under substantial violations of modeling assumptions.

\end{abstract}

\section{Introduction}

Online Randomized Controlled Experiments (A/B Tests) have become a cornerstone of decision-making in technology companies, commonly used to support deployment decisions and to evaluate their impact to the business~\cite{barajas2021online}. Large-scale platforms run thousands of concurrent experiments, randomly assigning units (e.g., users, devices) to treatments to establish causal links, such as measuring how recommender improvements affect engagement or revenue~\cite{Fiez2024Anduril}.

These experiments use a range of statistical estimation and inference techniques across multiple applications, including:
\begin{itemize}
\item{{\bf Launch Decisions}: Experimenters assess feature, product, or infrastructure changes to inform launch decisions.}
\item{{\bf Individual Effect Estimates:} Experimenters estimate the impact of proposed changes on relevant business metrics.}
\item{{\bf Cumulative Effect Estimates:} Experimenters estimate the cumulative effect of product changes (e.g., by business unit per year) to prioritize investments and inform resource allocation.}
\end{itemize}

Despite differences in experimental designs and processes across business units, decision quality ultimately depends on high-quality inference across A/B Tests~\cite{kohavi2020trustworthy}. A challenge arises when the same experiments are used both for selection (launch decisions) and for estimation. Conditioning inference on passing a selection threshold (i.e., selecting experiments based on noisy estimated effects) can cause ideas that are ineffective or only marginally effective to be declared ``winners'', inducing upward-biased effect estimates~\cite{xu2011bayesian} and miscalibrated confidence intervals among the selected experiments~\cite{andrews2024inference, xu2013MABbias}. This phenomenon is known as the Winner’s Curse.

The size of this selection bias typically scales directly with the magnitude of sampling variability and the decision threshold, and inversely with the true effect~\cite{van2021significance}.
While large-scale replication studies can overcome the Winner's Curse, they're prohibitively expensive. Scalable estimation and inference approaches under selection have thus gained significant industry attention~\cite{lee2018winner, ejdemyr2024estimating, kessler2024overcoming}.

This work presents a \textit{Bayesian Hybrid Shrinkage} (BHS) method that uses empirically motivated priors to resolve the Winner's Curse, providing unbiased estimates and credible intervals with good frequency properties. To address prior sensitivity concerns, we incorporate experiment-specific ``local shrinkage'' factors and evaluate our approach across simulated scenarios that reflect real-world assumption violations. For scalable deployment, we show how to infer the posterior without numerical integration and propose ongoing monitoring as a system health check. A public repository containing the implementation and code to replicate our simulation studies is available  \href{https://github.com/facebookresearch/bites/tree/main}{here}

We make three contributions to the research on Winner's Curse in online experimentation. First, BHS introduces experiment-specific local shrinkage factors that balance global learning with local adaptation. Second, we formalize the conditions under which Bayesian methods achieve well-calibrated inference under selection. Third, we provide a scalable approach for production experimentation platforms with demonstrated robustness to model misspecification.

The remainder of the paper is organized as follows. Section 2 reviews related work. Section 3 presents the problem and the BHS approach. Section 4 evaluates performance in simulations with varying assumption violations. Section 5 reports results on experimental data from Meta Platforms across two different business units, with a brief overview of deployment strategies. Section 6 concludes with future research directions.

\section{Related Work}
Researchers across disciplines—including genetics~\cite{zollner2007overcoming, ferguson2013empirical}, economics~\cite{lind1991winner, thaler1992winner, andrews2024inference}, and statistics~\cite{woody2022optimal, rasines2022empirical}—have increasingly recognized that selection under the Winner's Curse in experimentation can lead to unreliable conclusions, contributing to the \textit{Reproducibility and Replicability Crisis} and highlighting the need for more robust analytical methods~\cite{gupta2019top}.

In the Frequentist literature, researchers have addressed this problem with both theoretical~\cite{van2021significance} and practical~\cite{gelman2014beyond} approaches. Classical solutions, like those from \citet{benjamini2020selective}, focus on valid confidence intervals and p-values under selection. Methods such as \citet{andrews2024inference} use data splitting to achieve unbiased estimates and valid intervals, but these work best with large sample sizes and may not suit smaller studies.

More recently, \citet{lee2018winner} introduced a Frequentist estimator for online experiments that measures the average treatment effect aggregated across multiple experiments, using plug-in estimators and bootstrap methods to ensure unbiasedness and well-calibrated confidence intervals. While effective for aggregated effects (under their long-term properties), our approach extends these benefits to individual experiments.

In contrast to the above, a common Bayesian view on the Winner’s Curse~\cite{xu2011bayesian, rasines2022bayesian}, first argued by \citet{dawid1994selection}, is that no special adjustment for selection is needed since the posterior already reflects the observed data. However, later research \cite{yekutieli2012adjusted, woody2022optimal, mandel2007statistical} shows this only holds under certain assumptions. In practice, selection can interact with parameters in complex ways, like selective reporting or adaptive experiments. Whether adjustment is needed depends on the specific scenario, which we clarify for online experiments in section~\ref{section:theory}.

Within this literature, \textit{Empirical Bayes} (EB) shrinkage methods have emerged as a powerful framework for large-scale inference problems since Efron's seminal work~\cite{efron2012large}. These methods estimate the prior distribution from the data itself, enabling principled shrinkage of individual estimates toward a central value~\cite{nabi2022EB}. In the context of the Winner's Curse, EB approaches offer a particular promise: by borrowing strength across multiple experiments, they can mitigate the upward bias introduced by selection while maintaining computational tractability.

EB methods are focused primarily along two relevant dimensions. \textit{Global} shrinkage approaches~\cite{stephens2017false} assume exchangeability (i.e., apriori experiment effects have a joint-distribution that is invariant to their ordering) and apply shrinkage based on a singular estimated prior. This is computationally efficient but potentially overly aggressive when heterogeneity exists across experiments, a common scenario in online experimentation where product areas, user segments, or designs may exhibit different effect size distributions. \textit{Hierarchical} shrinkage methods~\cite{woody2022optimal} introduce additional structure to accommodate such heterogeneity and provide robustness to model misspecification, but typically require substantial computational overhead that is challenging to deploy at the scale of modern experimentation platforms processing thousands of experiments continuously.

Industry applications employ EB approaches in online experimentation, targeting aggregated effects across experiment collections. \citet{kessler2024overcoming} proposes a global shrinkage approach where experiments are considered exchangeable under a conjugate prior with shared hyperparameters. This is suitable for large-scale deployment, but potentially vulnerable to heterogeneity (albeit robustness under some forms of misspecification is studied in simulation). \citet{ejdemyr2024estimating} employ a hierarchical model over nested experiment collections via MCMC \cite{geyer1992practical}, offering greater accuracy and model flexibility at the cost of computational complexity. Our approach is closest to theirs, but with novel deployment-focused elements.

Our key contribution lies in bridging global and hierarchical EB approaches through experiment-specific local shrinkage factors. Our method estimates a global shrinkage parameter from historical data, then modulates it at the experiment level through local factors that adapt to observed effect magnitudes, providing robustness to prior misspecification. This \textit{hybrid} approach admits closed-form posterior inference without MCMC sampling, enabling scalable deployment while maintaining the flexibility of hierarchical methods.

This hybrid strategy addresses a critical gap in the EB literature: the tension between computational tractability (global approaches) and robustness to misspecification (hierarchical approaches).
We demonstrate the superiority of our approach through large-scale empirical evaluations and simulation studies that examine hyper-parameter misspecification, hidden selection, and heavy-tailed distributions.

\section{Formulation} \label{section:formulation}
Let $i=1,\ldots, N$ denote the total collection of experiments, with $j = 1, \ldots, m_i$ experimental units within each. 
An experimental unit refers to the unit of randomization, and is assigned (randomly) to a treatment condition $Z_{j} \in \{1, 0\}$. The outcome of interest is some business metric $Y_{j}$ for that unit in that experiment.

Within the potential outcomes framework~\cite{rubin1998more},
the target of inference (the estimand) is defined for each experiment as $\theta_{i} = \frac{E[Y(1)]}{E[Y(0)]}$, the ratio of \textit{true} averages between treated and control units in experiment $i$. Relative changes such as these are often relied on to help communicate results in business applications~\cite{Radwan2024Eval}. Then the \textit{Face Value} (FV) Frequentist estimator over the $m_{i}$ units is denoted $\hat{\theta}_{i}^{FV}$ and characterized as
\begin{align} \label{eq:facevalue}
\hat{\theta_{i}}^{FV} = \frac{\frac{1}{\sum_{j=1}^{m_i}Z_{j}}\sum_{j=1}^{m_i}Y_{j}Z_{j}}{\frac{1}{\sum_{i=1}^{m_i}(1 - Z_{j})}{\sum_{j=1}^{m_i}Y_{j}(1 - Z_{j})}},
\end{align}
 with standard error $\hat{\sigma_{i}}$. In aggregated settings like meta-experiments, our target of inference is typically $\Theta = \sum_{i=1}^{N}\theta_i$ that is estimated with $\hat{\Theta}^{FV} = \sum_{i=1}^{N}\hat{\theta_{i}}^{FV}$.

\subsection{Assumptions} \label{section:assumptions}
This section outlines our formulation's assumptions and their practicality in applied settings, particularly for aggregated inference targets. In our experimental setting, \textit{strong ignorability}~\cite{ding2018causal} is assumed and is guaranteed to hold under randomized treatment assignment.

When aggregating experiments, we commonly assume \textit{additivity} of treatment effects. This assumption is widely adopted since experiments run independently and effects are typically small on a relative basis. While effects are multiplicative, additivity provides a reasonable approximation for small magnitudes, a standard practice in the online experimentation literature~\cite{lee2018winner}.

We assume launches occur only with meaningfully \textit{positive} results, typical for organizations and business units optimizing a single metric (e.g., revenue). This assumption explains why the direction of bias is positive: while $\hat{\theta_{i}}^{FV}$ is unbiased for $\theta_{i}$, this does not hold when conditioned on $\left|\hat{\theta_{i}}^{FV}\right| > c$, where $c$ is some positive constant. The size of the bias increases in $c$ and the sampling variance of the estimator, $\mathbb{V}\left[\hat{\theta_{i}}^{FV}\right]$, but decreases in $|\theta_{i}|$~\cite{van2021significance}. While fewer, better-designed experiments could help, the constraint this introduces on platform throughput is undesirable given experimentation's critical role in innovation. \citet{andrews2024inference} face similar challenges in their reliance on 
data-splitting to mitigate selection, as this approach requires sufficiently large experiments 
to maintain statistical power after splitting, thereby limiting throughput as well.

A more practical solution is to discount the observed estimate to counterbalance the overestimation bias and reflect the true effect more closely. The Bayesian paradigm provides a principled framework for this adjustment.

\subsection{Bayesian Hybrid Shrinkage for Inference Under Selection} \label{section:model}

Based on the characterization in section~\ref{section:assumptions}, we propose the following BHS Model which can be formulated as a post-hoc adjustment:

\begin{align}
    \hat{\theta_{i}}|\theta_{i}, \hat{\sigma}^{2}_{i} &\sim \mathrm{N}(\theta_{i}, \hat{\sigma}^{2}_{i}),\\
    \theta_{i} | m_{0}, \lambda_{i}, \tau &\sim \mathrm{N}(m_{0}, \lambda_{i}\cdot\tau),\\
    \lambda_{i}|a, b &\sim \mathrm{Inverse-Gamma}\left(\frac{a}{2}, \frac{b}{2}\right).
\end{align}

The model contains parameters for flexibility across settings: $\tau$ is a global shrinkage parameter controlling information sharing across experiments (e.g., within product families), while $\lambda_i$ is a \textit{local} shrinkage parameter modulating experiment-level behavior. Parameters $a$ and $b$ are selected to reflect different prior properties for $\theta_{i}$ (e.g., heavy-tail behavior). Together, these parameters enable corrections for the Winner's Curse in diverse settings with strategy for inference over individual experiments as well as aggregated meta-experiments in section~\ref{section:inference}.

\subsection{Inference Strategy for Bayesian Hybrid Shrinkage}\label{section:inference}
Inference is either analytical or via posterior simulation using iterative sampling between the target and nuisance parameters\footnote{Nuisance parameters are parameters like $\lambda_i$ that are not targets of inference, but are required for its inference.}. The posterior distribution over $\theta_{i}$ is given as:

\begin{align}
\theta_{i}|\hat{\theta_{i}}, \lambda_{i}, \tau \sim \mathrm{N}\left(\frac{\hat{\sigma}^{2}_{i}}{\hat{\sigma}^{2}_{i} + \lambda_{i}\tau}m_{0} + \frac{{\lambda_{i}\tau}}{\hat{\sigma}^{2}_{i} + \lambda_{i}\tau}\hat{\theta_{i}}, \left(\frac{1}{\hat{\sigma}^{2}_{i}} + \frac{1}{\lambda_{i}\tau}\right)^{-1}\right). \label{eq:hs}
\end{align}

A special case of this posterior with $\lambda_{i} = 1$ corresponds to a Bayesian estimator that imposes global shrinkage as considered by~\citet{kessler2024overcoming}. We evaluate its performance later in this paper.

For aggregated results under additivity, inference is given by summing over the posteriors of the individual experiments:

\begin{align}
    \Theta | \hat{\Theta} \sim \mathrm{N}\left(\sum_{i=1}^{N}\left(\frac{\hat{\sigma}^{2}_{i}}{\hat{\sigma}^{2}_{i} + \lambda_{i}^{*}\tau}m_{0} + \frac{{\lambda^{*}_{i}\tau}}{\hat{\sigma}^{2}_{i} + \lambda_{i}^{*}\tau}\hat{\theta}_{i}\right), \sum_{i=1}^{N}\left(\frac{1}{\hat{\sigma}^{2}_{i}} + \frac{1}{\lambda_{i}^{*}\tau}\right)^{-1}\right),
\end{align}
which converges in distribution to the estimand as follows.

\begin{theorem}
Let $\Theta = \sum_{i=1}^{N} \theta_i$ denote the true aggregated treatment effect, and let $\pi(\Theta \mid \boldsymbol{\hat{\theta}}_m)$ denote the posterior distribution of this aggregate effect given the vector of observed estimates $\boldsymbol{\hat{\theta}}_m = (\hat{\theta}_{1m}, \dots, \hat{\theta}_{Nm})$ from experiments with sample size $m$. As $m \to \infty$, the posterior distribution $\pi(\Theta \mid \boldsymbol{\hat{\theta}}_m)$ converges in distribution to the true value $\Theta$.
\end{theorem}

\begin{proof}

By the Bernstein-von Mises theorem, for each independent experiment $i$, with sample size $m$, the posterior distribution $\pi(\theta_i \mid \hat{\theta}_{i_m})$ converges to a Normal distribution centered at the true parameter $\theta_i$ with posterior variance approaching zero. Invoking the Continuous Mapping Theorem, as the marginal posteriors converge in probability to their respective true parameters, the posterior distribution of the sum, $\pi(\Theta \mid \boldsymbol{\hat{\theta}}_m) \to_{d} \Theta$.
\end{proof}

As inference over the estimand requires the nuisance parameter $\tau$, we use a plug-in estimate inferred from past experiments. Assume that we have a curated collection of experiments with point estimates $\hat{\eta}_{k}$ modeled as:
\begin{align}
\hat{\eta}_{k}|\eta, \gamma\sim\mathrm{N}(\eta, \gamma),
\end{align}
\noindent with parameter $\gamma$ treated as known via approaches such as meta-analysis~\cite{howes2024understanding}. We specify a prior over $\eta$ that facilitates information sharing via pooling as: 

\begin{align}
\eta|\tau\sim\mathrm{N}(0, \tau).
\end{align}

The model over the curated experiment effects and the prior yields a marginal likelihood that we optimize over to estimate $\tau$:

\begin{equation}
\tau^* = \mathrm{argmax}_{\tau} \int \mathcal{L}(\hat{\eta_{k}}|\tau) \,\mathrm{d}\eta \;=\; \mathrm{argmax}_{\tau} \;\mathrm{N}(\hat{\eta}_{k}|0, \tau + \gamma).
\end{equation}
Once this parameter has been set, local shrinkage can be modulated via the full conditional posterior distribution over $\lambda_{i}$:

\begin{align}
\lambda_{i}|\theta_{i}, \tau\sim\mathrm{Inverse-Gamma}\left(\frac{a+1}{2}, \frac{(\theta_{i} - m_0)^2 + b \tau}{2 \tau}\right).
\end{align}
This full conditional posterior has a closed form conditional on $\theta_{i}$ and $\tau$ which necessitates iterative sampling between them for inference over $\lambda_{i}$. Despite this iterative sampling step being faster than MCMC, sampling at scale can nevertheless be challenging. We therefore replace $\theta_{i}$ with its sufficient statistic $\hat{\theta}_{i}^{FV}$ and $\tau$ with $\tau^{*}$. Conditional on these parameters being set, we fix $\lambda_i$ at its posterior mode, adopting weakly informative hyperparameters $a=b=1$:

\begin{align}
\label{eq:11}
\lambda_{i}^{*} = \frac{(\hat{\theta}_{i} - m_0)^2 + \tau^*}{4\tau^*}.
\end{align}

By fixing the value of $\lambda_{i} = \lambda_{i}^{*}$, the following property holds for the BHS estimator generally:

\begin{theorem}
Under fixed values of $\lambda_{i}^{*}$, $\tau^{*}$ the posterior mean is a consistent estimator for $\theta_{i}$.
\end{theorem}

\begin{proof}
Let $\hat{\sigma}^{2}_{i_m}$ denote the variance of the sampling distribution of $\hat{\theta}_{i}$ based on $m$ units of observation for an unbiased estimator $\hat{\theta}_{i}$. As $m\rightarrow \infty$, then $\hat{\sigma}^{2}_{i_m} \rightarrow 0$ and $\frac{\hat{\sigma}^{2}_{i_m}}{\hat{\sigma}^{2}_{i_m} + \lambda_{i}^{*}\tau}m_{0} + \frac{{\lambda^{*}_{i}\tau}}{\hat{\sigma}^{2}_{i_m} + \lambda_{i}^{*}\tau}\hat{\theta_{i}}\rightarrow \theta_{i}$.
\end{proof}

In practice, we set $m_{0} = 0$, which based on Proposition 3 in~\cite{van2021significance} typically yields lower bias than the Frequentist estimator. $\hat{\sigma}^{2}_{i}$ is the standard plug-in for the variance of the estimator $\hat{\theta}_{i}^{FV}$. We set $a=b=1$ to induce a diffuse Cauchy prior over $\theta$. Under these parameters, the following property also holds generally:

\begin{theorem}
\label{th:red}
Let $\hat{\theta}_i \mid \theta_i \sim \mathcal{N}(\theta_i,\hat{\sigma}_i^2)$. Consider the BHS estimator where the global variance $\tau=\tau^*$ is fixed and the local shrinkage parameter is set to its conditional posterior mode $\lambda_i=\lambda^{*}_{i}$ with hyper-parameters $a=b=1$ and substituting $\theta_i$ with $ \hat{\theta}_i$ as in Equation~\ref{eq:11}. Then $\hat{\theta}_i^{BHS}$ satisfies the redescending property:
\begin{equation}
\lim_{|\hat{\theta}_i|\to\infty}\left| \hat{\theta}_i^{BHS} - \hat{\theta}_i \right| \;=\; 0.
\end{equation}

\end{theorem}

\begin{proof}
Using the standard Normal-Normal update with $m_0=0$ and effective prior variance $v_i := \lambda_i^*\tau^*$, the posterior mean is:
\begin{equation}
\hat{\theta}_i^{BHS}
=
\frac{v_i}{v_i+\sigma_i^2}\,\hat{\theta}_i.
\end{equation}
The magnitude of the shrinkage (difference relative to face value) is:
\begin{equation}
\left| \hat{\theta}_i^{BHS} - \hat{\theta}_i \right|
=
|\hat{\theta}_i| \left( 1 - \frac{v_i}{v_i+\sigma_i^2} \right)
=
|\hat{\theta}_i|\frac{\sigma_i^2}{v_i+\sigma_i^2}.
\end{equation}
Substituting the plug-in value $v_i = \lambda_i^*\tau^* = \frac{\hat{\theta}_i^{\,2}+\tau^*}{4}$:
\begin{equation}
\left| \hat{\theta}_i^{BHS} - \hat{\theta}_i \right|
=
\frac{\sigma_i^2 |\hat{\theta}_i|}{\frac{\hat{\theta}_i^{\,2}+\tau^*}{4} + \sigma_i^2}
=
\frac{4\sigma_i^2|\hat{\theta}_i|}{\hat{\theta}_i^{\,2} + (\tau^*+4\sigma_i^2)}.
\end{equation}
As $|\hat{\theta}_i|\to\infty$, the numerator grows linearly while the denominator grows quadratically. Thus, the expression approaches 0.
\end{proof}

\paragraph{Contrast with Global Shrinkage}
Under Global Shrinkage, $\lambda_i^* \equiv 1$, implying a fixed prior variance $v_i=\tau^*$. The shrinkage is:
\begin{equation}
\left| \hat{\theta}_i^{Global} - \hat{\theta}_i \right|
=
|\hat{\theta}_i|\frac{\sigma_i^2}{\tau^*+\sigma_i^2}.
\end{equation}
This function grows linearly with $|\hat{\theta}_i|$ and is unbounded. In contrast, the BHS plug-in estimator's bias vanishes for extreme observations, ensuring that the strongest signals are preserved.

\subsection{Modeling Selection in Bayesian Analyses}\label{section:theory}
The Winner's Curse occurs when selection alters the sampling distribution of the estimator. Frequentist methods that correct for this require explicit selection model adjustment~\cite{andrews2024inference}. For Bayesians, this is less clear: \citet{dawid1994selection} argues that posteriors already condition on observed data, making a selection model irrelevant. However, this may not account for the selection mechanism.

In Bayesian analysis, whether explicit adjustment is needed depends on how parameters and data interact with selection~\cite{yekutieli2012adjusted, woody2022optimal}. If parameters and data are jointly sampled, no adjustment is required. Conversely if parameters are fixed from their marginal distribution with data sampled conditionally, adjustment is needed. For example:

\begin{itemize}
\item{Experiments testing different product changes. Data is used to make launch decisions via predetermined criteria, and we estimate effect sizes for launched experiments. Here, a Bayesian approach does \textit{not} require selection adjustment.}
\item{A single proposed change to a recommender system, tested multiple times across experiments. Using only experiments that pass launch criteria to estimate the change's magnitude \textit{would} require selection adjustment.}
\end{itemize}

The following theorem describes these properties generally.

\begin{theorem} \label{theorem:jointselection}
With joint selection, the selective and unadjusted posteriors are equivalent, eliminating the need for a selection model.
\end{theorem}

\begin{proof}
Under joint selection truncated to the selected sample $S$, let the joint distribution of $\hat{\theta}, \theta$ be $\pi_{S}(\hat{\theta}, \theta) = \pi(\theta)\cdot f(\hat{\theta}|\theta)\cdot \frac{1(\hat{\theta}\in S)}{\pi(S)}$, where the marginal selection probability is
$\pi(S)=\int\int_{S}\pi(\hat{\theta}|\theta) \pi(\theta) \allowbreak \mathrm{d}\hat{\theta} \mathrm{d}\theta$. The marginal density of the data under selection is $m(\hat{\theta}) = \int f(\hat{\theta}|\theta)\pi(\theta)\mathrm{d}\theta$. Within the selected sample, this marginal density is $m_S(\hat{\theta})=\int 1(\hat{\theta}\in S)\pi(\theta)\frac{f(\hat{\theta}|\theta)}{\pi(S)}\mathrm{d}\theta$ $=\frac{m(\hat{\theta})}{\pi(S)}\cdot 1(\hat{\theta}\in S)$. Next consider the posterior under selection $\pi_{S}(\theta|\hat{\theta}) = \frac{\pi_{S}(\theta, \hat{\theta})}{m_{S}(\hat{\theta})} = \frac{\pi(\theta)\cdot f(\hat{\theta}|\theta)}{m(\hat{\theta})} = \pi(\theta|\hat{\theta})$ which is identical to the unadjusted posterior.
\end{proof}
Many applied settings fit the joint distribution paradigm, where well-motivated priors suffice to overcome Winner's curse. Validating the modeling choice and ensuring prior calibration are vital.

\subsection{Validation via Predictive Checking}\label{section:validation}
Given the complexities of validation in experimental settings, we take inspiration from \textit{predictive checking}~\cite{rubin1998more, gelman1996posterior}. The method works by using the posterior distribution of model parameters $\theta$ to generate simulated "replicated" datasets $\hat{\theta}_{rep}$, which represent what the data should look like if the model $M$ was true. These simulated datasets are then compared to the actual observed data $\hat{\theta}$ using test statistics or graphical displays. If the real data appear unusual or extreme relative to the replicated data, this suggests potential model misspecification, thereby classifying this as an assessment of \textit{ goodness-of-fit}.

The computation required for this assessment is a straightforward byproduct of posterior predictive simulation. Performance can be benchmarked via any statistic of interest $T(\cdot)$. If $g(\cdot)$ denotes a comparison between $T(\hat{\theta})$ and $T(\hat{\theta}_{rep})$, then the following is a summary of goodness-of-fit:
\begin{align}
p[g(T(\hat{\theta}_{rep}), T(\hat{\theta}))|M, \hat{\theta}].
\end{align}

To operationalize this, an intuitive comparison is the tail area probability $p[g(T(\hat{\theta}_{rep}), T(\hat{\theta}))|M, \hat{\theta}] = p[T(\hat{\theta}_{rep}) \geq T(\hat{\theta})$]. It is an analogue to a p-value that can be computed for any model of choice averaged over the posterior distribution of the parameter of interest under the model. By redefining $g$, the same construction can be used for quantities like coverage for uncertainty intervals, or mean-squared error for accuracy. This concept can be augmented by data splitting strategies such as replication studies or data fission~\cite{leiner2023data}. We will demonstrate the usage of this technique in the empirical analyses performed.

\section{Simulation Studies}

To demonstrate the expected performance of different estimation approaches under different conditions, we evaluate them in simulation. We consider the case where the prior is correctly specified coupled with cases where the prior is misspecified in ways that we expect in real-world online experimentation platforms:
\begin{enumerate}
    \item{{\bf Misspecified mean:} The mean used for analysis is different from the mean parameter used to generate the true effects in the simulation.}
    \item{{\bf Heavy-tailed distributions}: The distribution that generates the effects in simulation has heavier tails than the prior used for analysis.}
    \item{{\bf Hidden selection:} The distribution used to generate the true effects is a correlated bivariate distribution. The selection criteria applies to both variables, but analysis only considers the target of inference.}
\end{enumerate}

In each scenario, we simulate a collection of experiments $i = 1, \dots, N$. We compare three distinct estimation strategies based on the formulation in Section~\ref{section:formulation}.
\begin{itemize}
    \item \textbf{Frequentist Face Value ($\hat{\theta}_{i}^{FV}$):} The conventional Frequentist ratio estimator (Eq.~\ref{eq:facevalue}), which uses only the observed data from the current experiment without adjustment.
    \item \textbf{Bayesian Global Shrinkage ($\hat{\theta}_{i}^{G}$):} The posterior mean of a Bayesian model where the shrinkage factor is shared globally across all experiments (Eq.~\ref{eq:hs}  with $\lambda_i^*=1$).
    \item \textbf{Bayesian Hybrid Shrinkage ($\hat{\theta}_{i}^{H}$):} Our proposed estimator (Eq.~\ref{eq:hs}), which utilizes the posterior mode of the local shrinkage parameter $\lambda_i^*$ to modulate the global shrinkage based on experiment-specific evidence.
\end{itemize}

Let $\hat{\theta}_{i}^{k}$ denote the point estimate for the true effect $\theta_i$, where $k \in \{FV, G, H\}$ indexes the estimation approach. We define the estimation error as $\delta_{i}^{k} = \hat{\theta}_{i}^{k} - \theta_i$. Conventional 90\% confidence intervals are used for the Face Value approach, and 90\% central credible intervals are used for the Bayesian approaches. We use the selection indicator $S_i \in \{0, 1\}$ to denote which experiments are selected for inference. The approaches are compared in terms of three performance metrics:
\begin{itemize}
    \item {\textit{Mean Squared Error}, $MSE^{k} =  \frac{1}{N} \sum_{i=1}^{N} (\delta^{k}_{i})^2$}
    \item \textit{Bias}, $\Delta^{k} = \frac{1}{N} \sum_{i=1}^{N} \delta^{k}_{i}$
    \item \textit{Coverage Probability} for 90\% confidence/credible intervals
\end{itemize}

\subsection{Inference Under A Correctly-Specified Prior} \label{section:correctprior}

We first consider a setting where effects are drawn from a normal distribution with the same parameters as those used for the analysis prior. We generate draws as:
\begin{align}
    \theta_{i} &\sim \mathrm{N}(\mu, \epsilon)
\end{align}
where $\mu$ and $\epsilon$ are constants.

For each experiment $i$, the face-value estimate $\hat{\theta}^{FV}_i$ of the effect $\theta_i$ is simulated by taking a draw from:
\begin{align}
    \hat{\theta}^{FV}_{i} &\sim \mathrm{N}(\theta_i, \hat{\sigma}_{i}^2)
\end{align}
where $\hat{\sigma_i}$ is the estimated standard error in each experiment, itself simulated by drawing from:
\begin{align}
    \hat{\sigma_i} &\sim \mathrm{N}(\sigma_i, \kappa)
\end{align}
where $\sigma$ and $\kappa$ are constants.

Experiments are selected for inference if the face-value estimate exceeds a threshold $\hat{\theta_i} > C_{i}$, where $C_{i}$ corresponds to the threshold in a standard t-test with a level of significance $\alpha = 0.05$. Other selection rules are also possible and would typically lead to similar selection effects.

\paragraph{\bf \textit{Global Shrinkage Eliminates Selection Bias in Expectation.}}
Fixing the values of $\mu$, $\epsilon$, $\sigma$, and $\kappa$, and setting $m_0 = \mu$ and $\tau = \epsilon$, the Face Value estimator is subject to a positive upward bias conditional on selection ($S_i=1$). The Bayesian estimator with global shrinkage, however, remains unbiased in expectation even when conditioning on selection, consistent with Theorem~\ref{theorem:jointselection}.

This calibration is verified in Figure~\ref{fig:histograms}, which shows the distribution of the error $\delta_{i}^{k}$ for $k \in \{FV, G\}$ across $N=10,000$ simulated experiments. For an unbiased estimator, the error distribution should be centered at zero. For the Bayesian estimator with Global Shrinkage this is always the case, irrespective of the selection threshold. The Face Value estimator on the other hand consistently demonstrates a bias which is entirely driven by the selection mechanism. Since the prior is correctly specified in this setting, BHS is excluded from the comparison.

\begin{figure}
    \centering
    \includegraphics[width=1.0\linewidth]{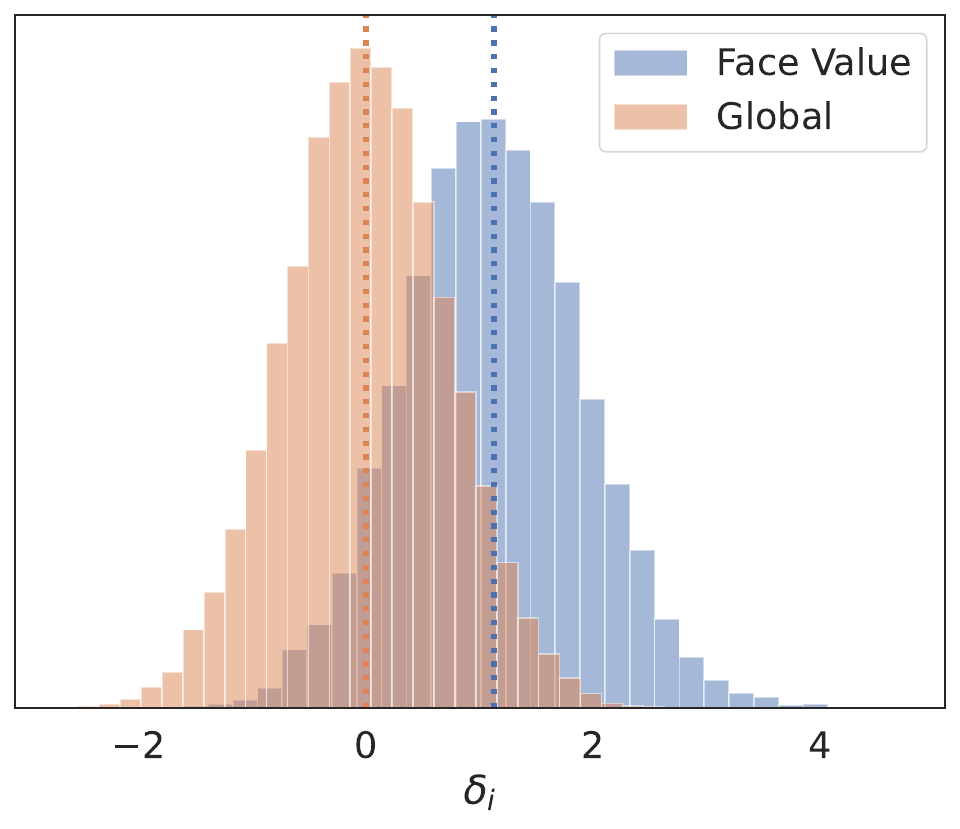}
    \caption{$\delta_i = \hat{\theta}_i - \theta_i$ across $N = 10,000$ simulated experiments for the Face Value Frequentist estimator and the Bayesian estimator with Global Shrinkage.}
    \label{fig:histograms}
\end{figure}

\paragraph{\bf \textit{The Bias Due to Selection Increases as Experiment Power Decreases}}

We now fix $\mu$, $\epsilon$ and $\kappa$, and compare the three estimators as a function of sampling variance in Figure~\ref{fig:correctly_specified}. For Face Value, the MSE increases and the coverage probability decreases as the sample variance increases. For the Bayesian estimator with global shrinkage, the MSE, bias and coverage probability remain optimal. When global shrinkage prior reflects the correct model exactly, the posterior converges to the correct estimand with minimal requirements on the data. Under this setting no other model is expected to perform better. For the BHS estimator, the MSE and coverage probability do increase and decrease respectively, but at a much slower rate than the Face Value estimator, reflecting Theorem~\ref{th:red}.

\begin{figure}
    \centering
    \subfloat{
        \includegraphics[width=0.65\textwidth]{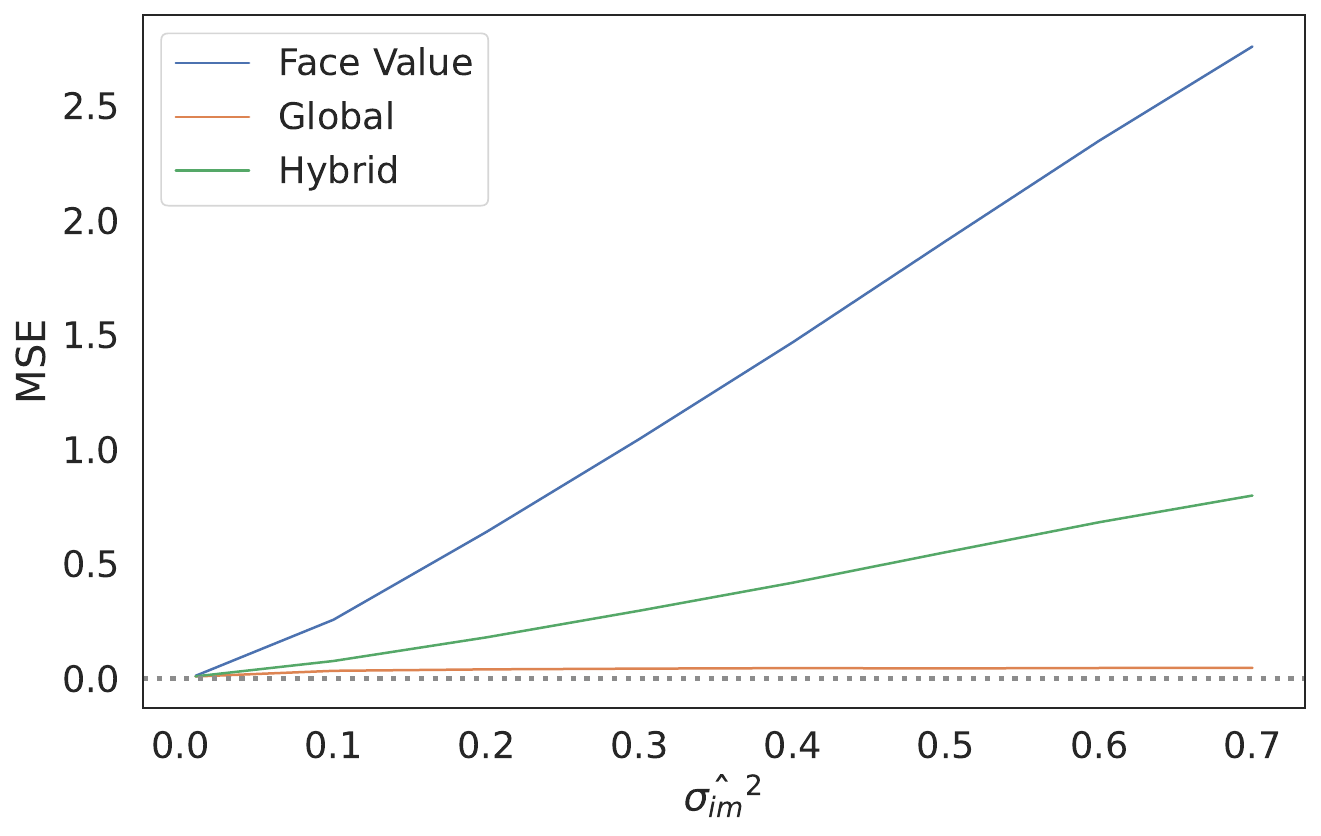}
        \label{subfig:mse_power}
    } \\
    \subfloat{
        \includegraphics[width=0.65\textwidth]{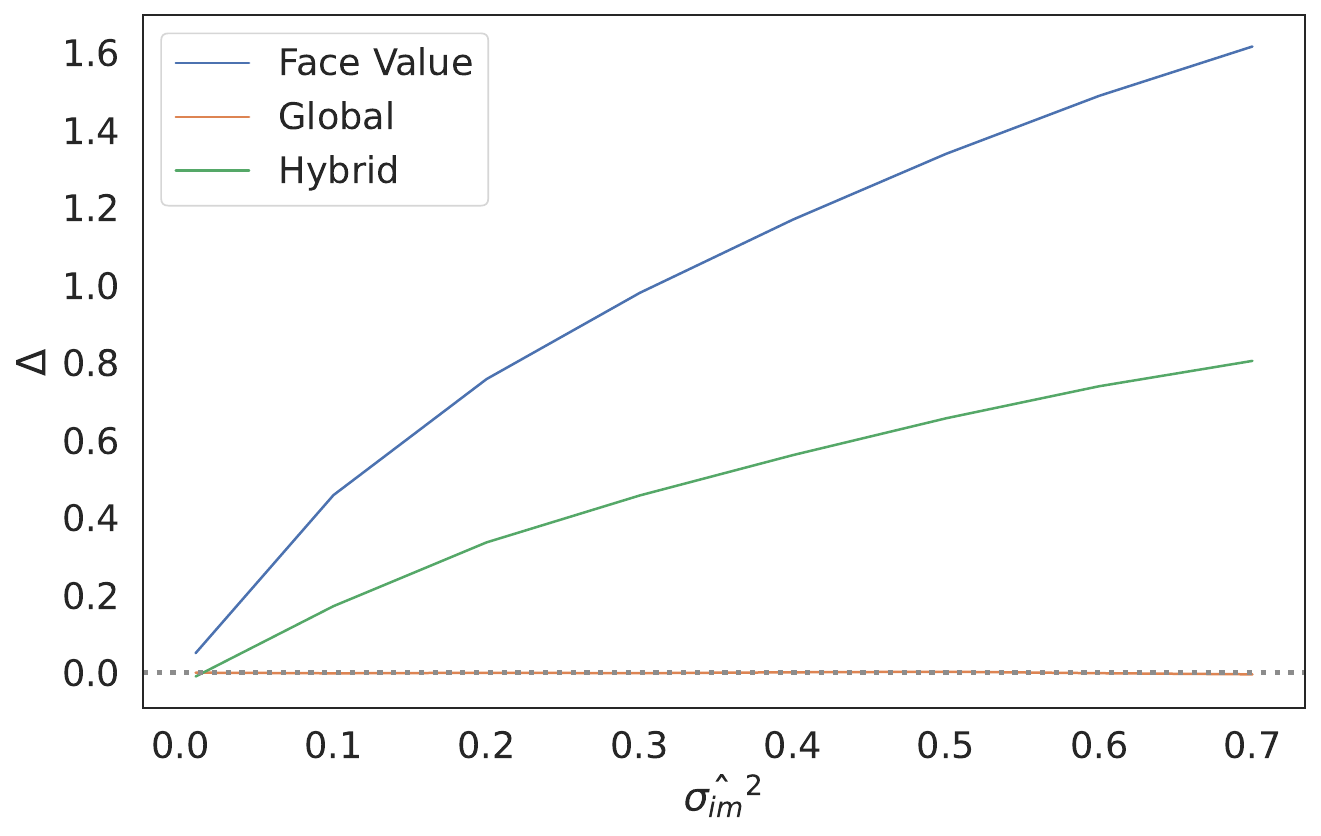}
        \label{subfig:bias_power}
    } \\
    \subfloat{
        \includegraphics[width=0.65\textwidth]{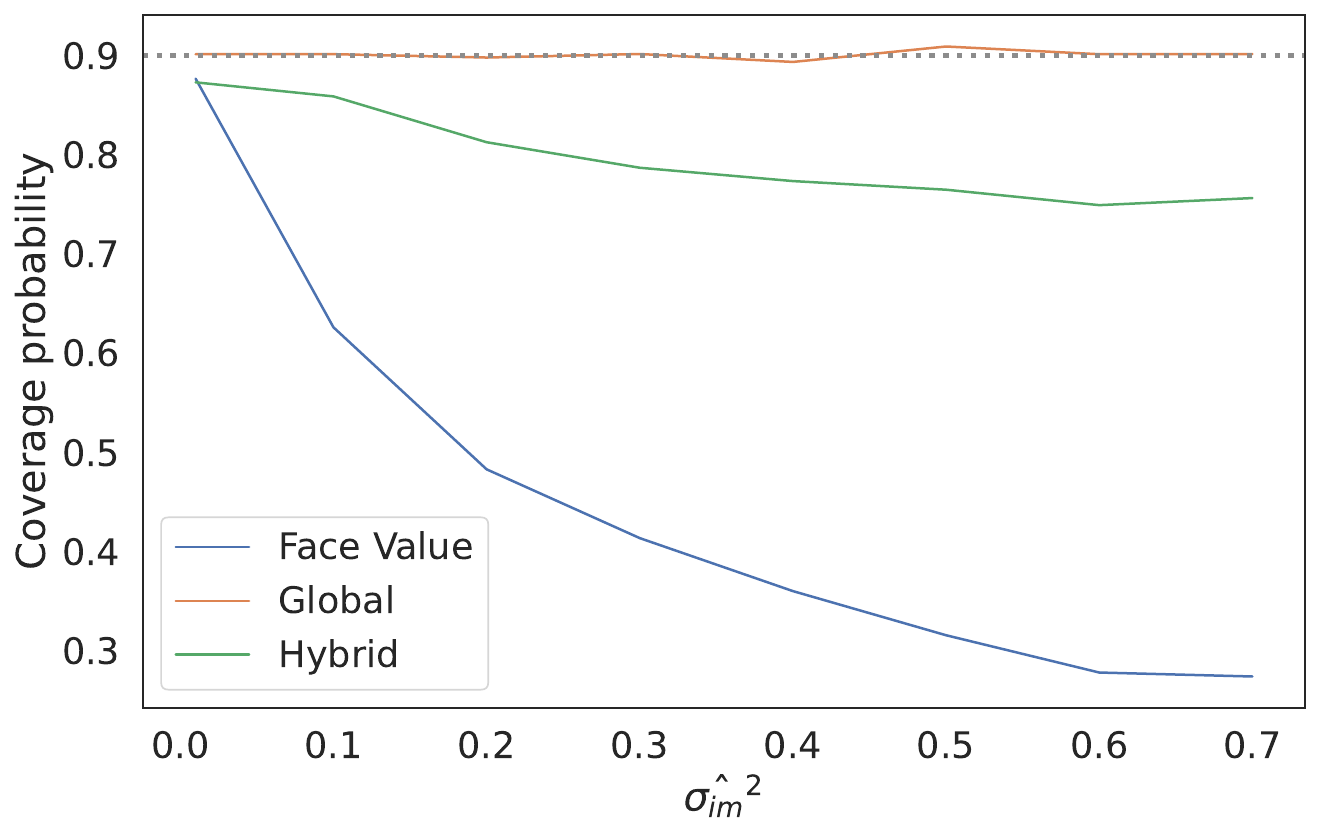}
        \label{subfig:coverage_power}
    }
    \caption{MSE (top), bias ($\Delta$, middle) and coverage probability (bottom) as a function of sampling variance for face value, Bayesian with global shrinkage and BHS approaches.}
    \label{fig:correctly_specified}
\end{figure}

\subsection{Inference Under A Misspecified Prior}

\subsubsection{Inference Under A Misspecified Prior Mean}

We now consider a setting where effects are drawn from a normal distribution, but the prior mean for analysis is misspecified. The simulation proceeds as in Section~\ref{section:correctprior}, but we set $m_0 = 0$ and vary $\mu$. Since effects are small in relevant applied settings, we expect $m_0 = 0$ to often be a practical choice as in~\cite{van2021significance}, but may induce more shrinkage than necessary if true effects are larger. See results in Figure~\ref{fig:misspecified}, top panel.

With increasing $\mu$, the performance of the Bayesian estimator with global shrinkage deteriorates, while the performance of the face-value estimator improves. This is intuitive: as the true effects grow larger relative to the sampling variance, the impact of selection on the classical estimator is expected to lessen. Conversely, as the true effects grow larger, assuming a prior mean $m_0$ with a singular variance leads to increasingly too-aggressive shrinkage.

The BHS estimator demonstrates superior properties to the face-value estimator across the range of prior mean values. When $\mu$ is small the Bayesian estimator with global shrinkage demonstrates better properties. But as the degree of misspecification increases, the hybrid approach proves significantly more robust across all performance metrics.

\subsubsection{Inference Under Heavy-Tailed Distributions}
While a Normal prior is an attractive choice in terms of computational efficiency, the approximation will suffer if the true distribution of effects is heavy-tailed. To examine robustness to this, we consider a setting where effects are drawn from a distribution with heavier tails than the distribution assumed by the analysis prior.

We simulate a collection of experiments from a t-distribution with $\nu$ degrees of freedom:
\begin{align}
    \theta_{i} \sim t_\nu(\mu, \epsilon).
\end{align}
Selection and estimation then proceed as described in section~\ref{section:correctprior}. These results are shown in the middle panel of Figure~\ref{fig:misspecified}.

As $\nu$ increases, t-distribution approaches the Normal distribution. The tail behavior diminishes and the prior assumption becomes more sensible. Hence for large $\nu$, the performance of the estimation approaches is similar to that seen in section \ref{section:correctprior}.

As $\nu$ decreases the tail behavior increases and the assumption of a Normal prior becomes less appropriate. In the extreme the shrinkage induced by the Bayesian estimator with global shrinkage is worse even than the face-value estimator. The BHS estimator is superior to the face-value estimator across all of the performance metrics for any value of $\nu$. For low $\nu$ it is a significant improvement against the Bayesian Global Shrinkage benchmark.

\subsubsection{Inference Under Hidden Selection}
Practitioners often consider additional information at selection time that is not subsequently used for inference~\cite{Li2027optimalPolicies}. A typical example is selecting a launch based on a vector of metrics (e.g., clicks and revenue), but only estimating the effect on one primary metric (e.g., revenue) for the final report.

To model this, we consider a setting where effects are drawn from a Bivariate Normal distribution $\theta_i \sim N_2(\mu, \Sigma)$, where $\theta_i = (\theta_{x,i}, \theta_{y,i})$ and $\Sigma$ allows for correlation $\rho$ between the dimensions. Selection is applied to both dimensions: an experiment is selected ($S_i=1$) only if \textit{both} estimates $\hat{\theta}_{x,i}$ and $\hat{\theta}_{y,i}$ exceed their respective thresholds. However, inference is performed solely on the target dimension $\theta_{y,i}$ using the univariate methods described in Section 4.1. These results are shown in the lower panel of Figure~\ref{fig:misspecified}.

As $\rho \to 0$, there is no correlation between $\theta_x$ and $\theta_y$, and the performance of each of the estimation approaches tends to that described in section~\ref{section:correctprior}. In the other direction, the performance of the Bayesian estimator with global shrinkage degrades. Though the performance of the face value approach improves, in terms of bias and MSE it remains inferior to the Bayesian global shrinkage approach, only improving on it in terms of coverage probability for high correlation.

BHS also outperforms the face-value approach in terms of MSE and bias across the full range of $\rho$. As $\rho$ increases, the Bayesian global shrinkage estimator leads to extreme shrinkage, while its hybrid counterpart outperforms it in terms of bias. The hybrid approach also provides good coverage properties across the range, though it becomes inefficient due to intervals overcover, such that the MSE remains inferior to the global approach.

\begin{figure*}
    \centering
    \subfloat{
        \includegraphics[width=0.3\textwidth]{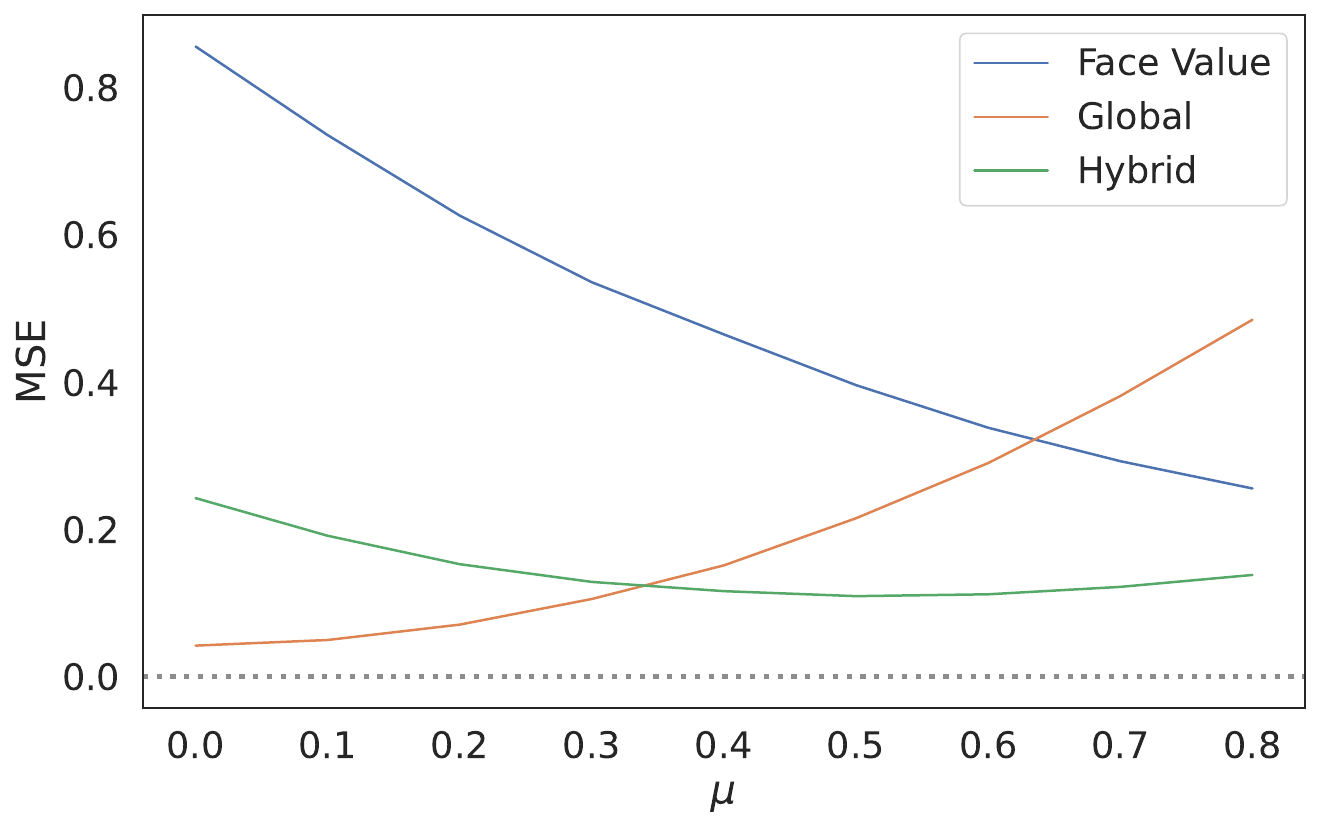}
        \label{subfig:mse_mu}
    }\hfill
    \subfloat{
        \includegraphics[width=0.3\textwidth]{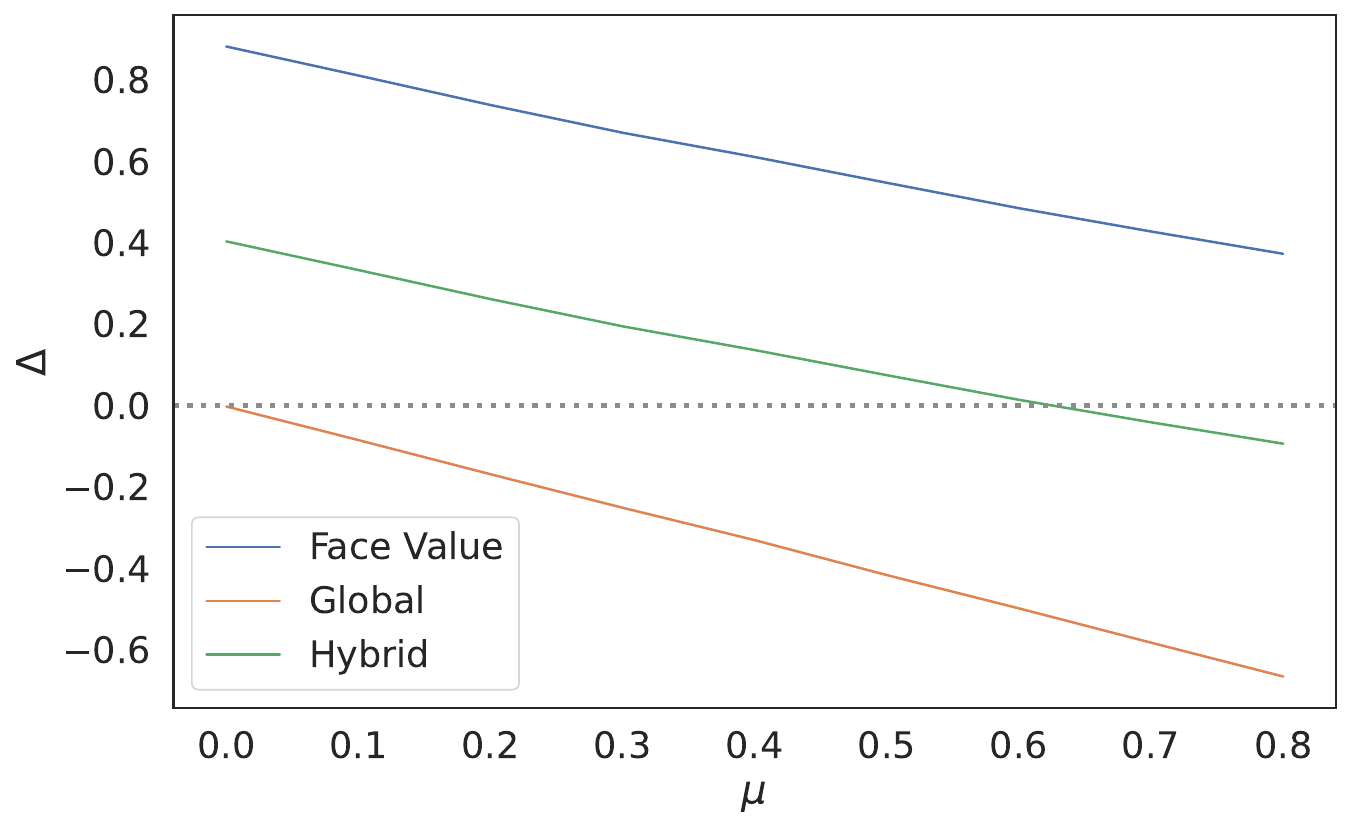}
        \label{subfig:bias_mu}
    }\hfill
    \subfloat{
        \includegraphics[width=0.3\textwidth]{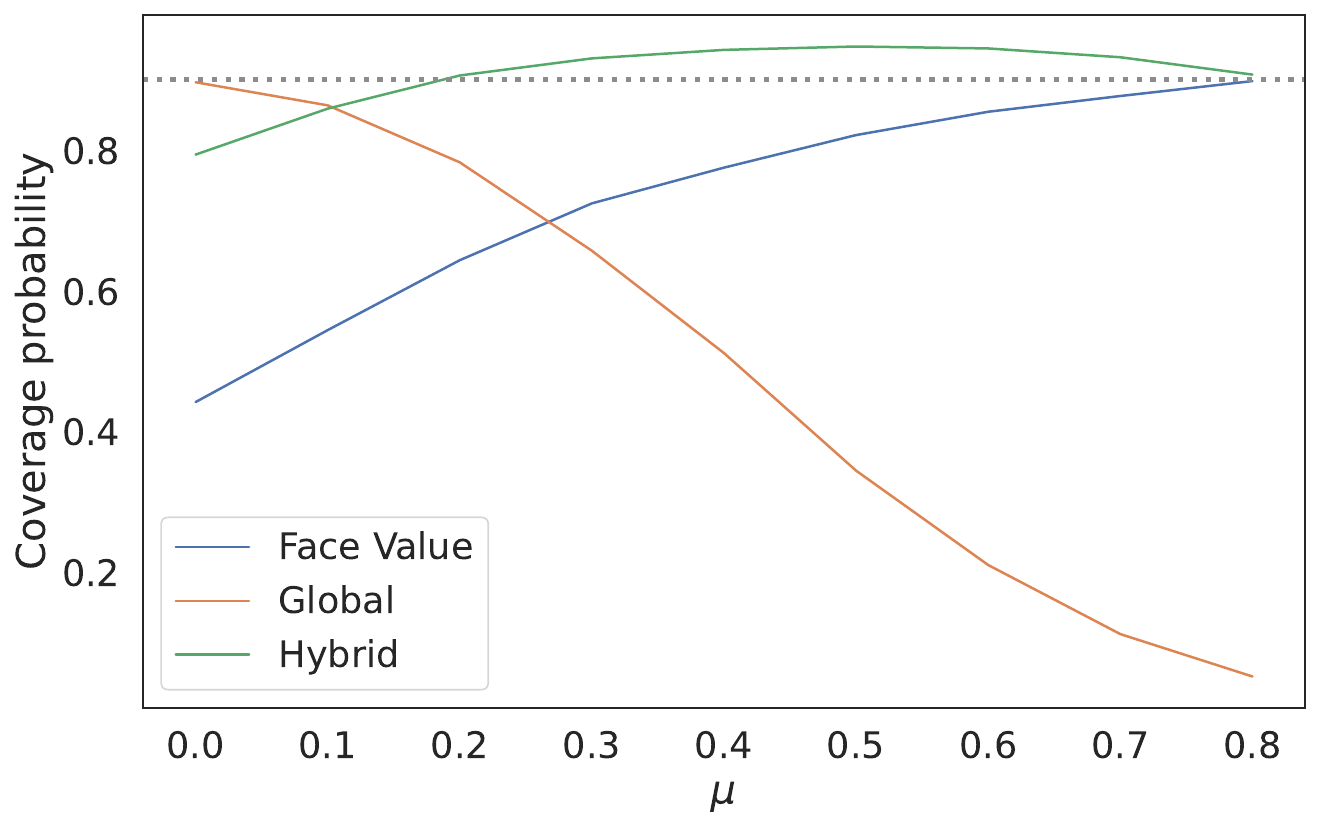}
        \label{subfig:coverage_mu}
    }\\
    \subfloat{
        \includegraphics[width=0.3\textwidth]{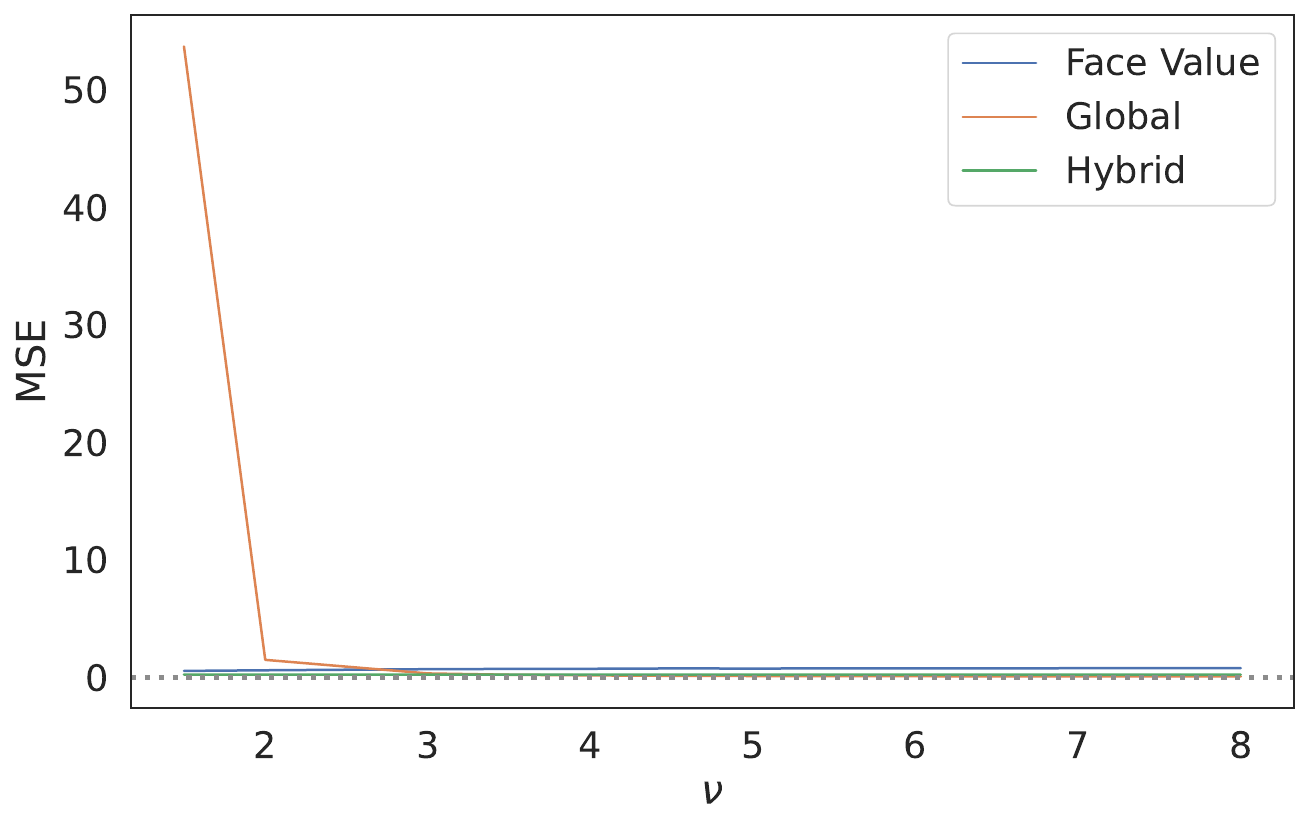}
        \label{subfig:mse_nu}
    }\hfill
    \subfloat{
        \includegraphics[width=0.3\textwidth]{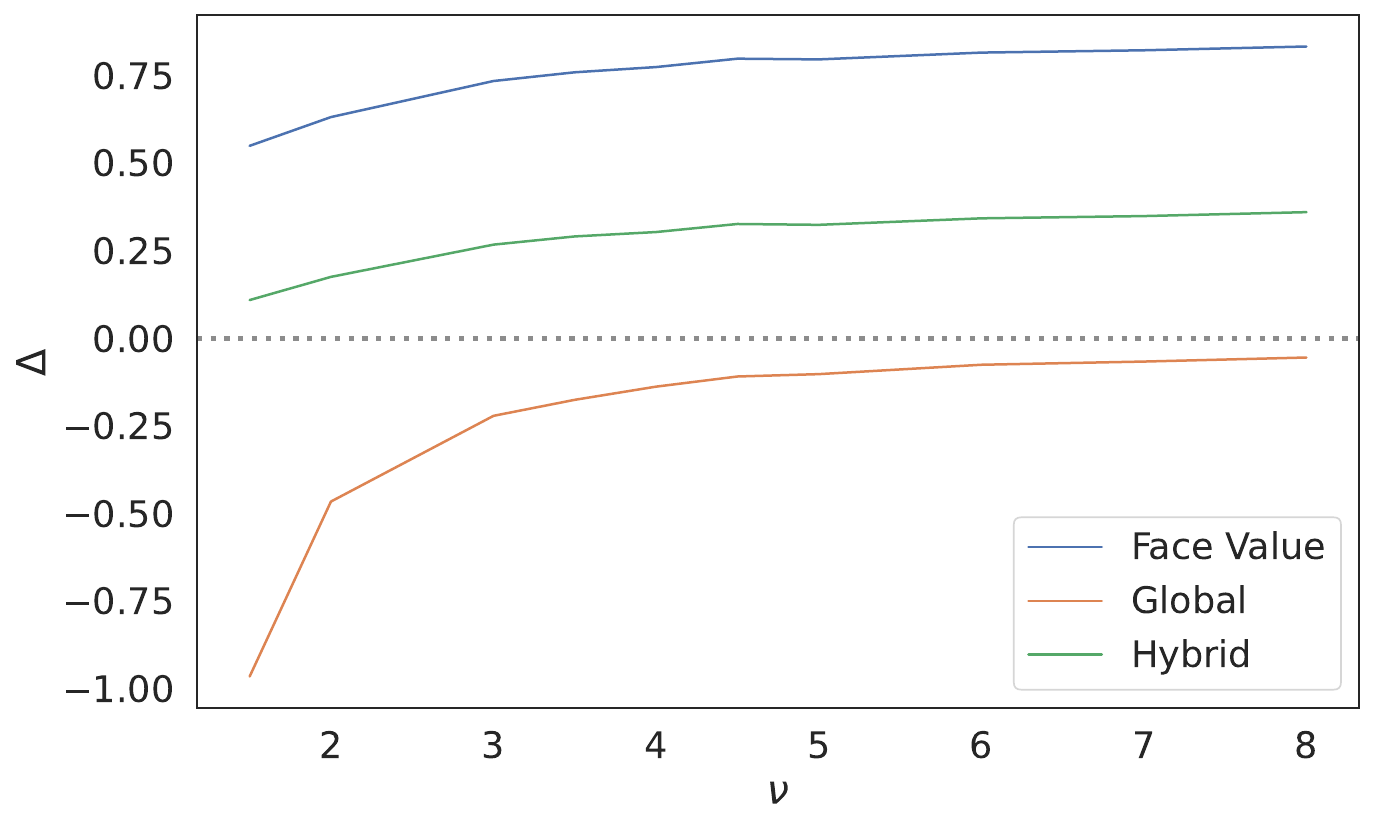}
        \label{subfig:bias_nu}
    }\hfill
    \subfloat{
        \includegraphics[width=0.3\textwidth]{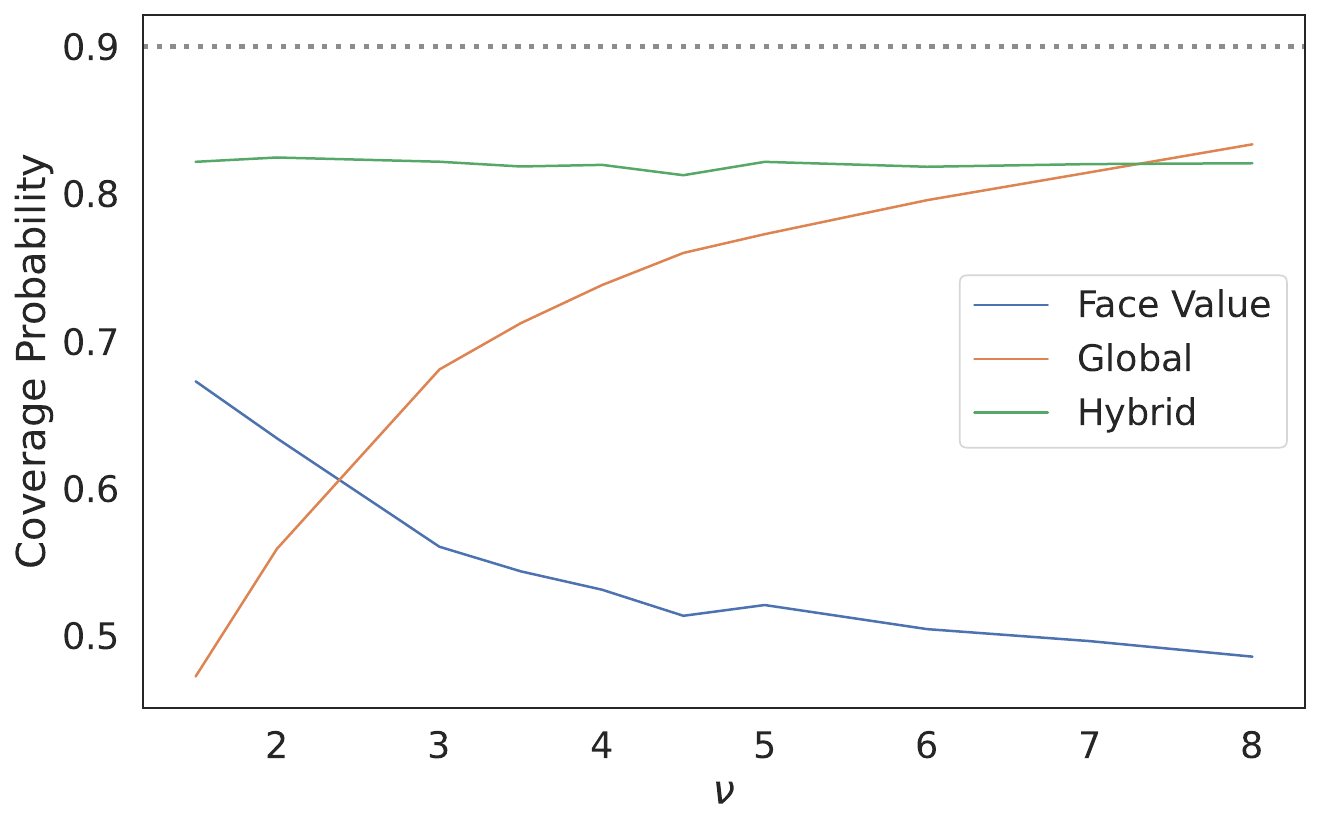}
        \label{subfig:coverage_nu}
    }\\
    \subfloat{
        \includegraphics[width=0.3\textwidth]{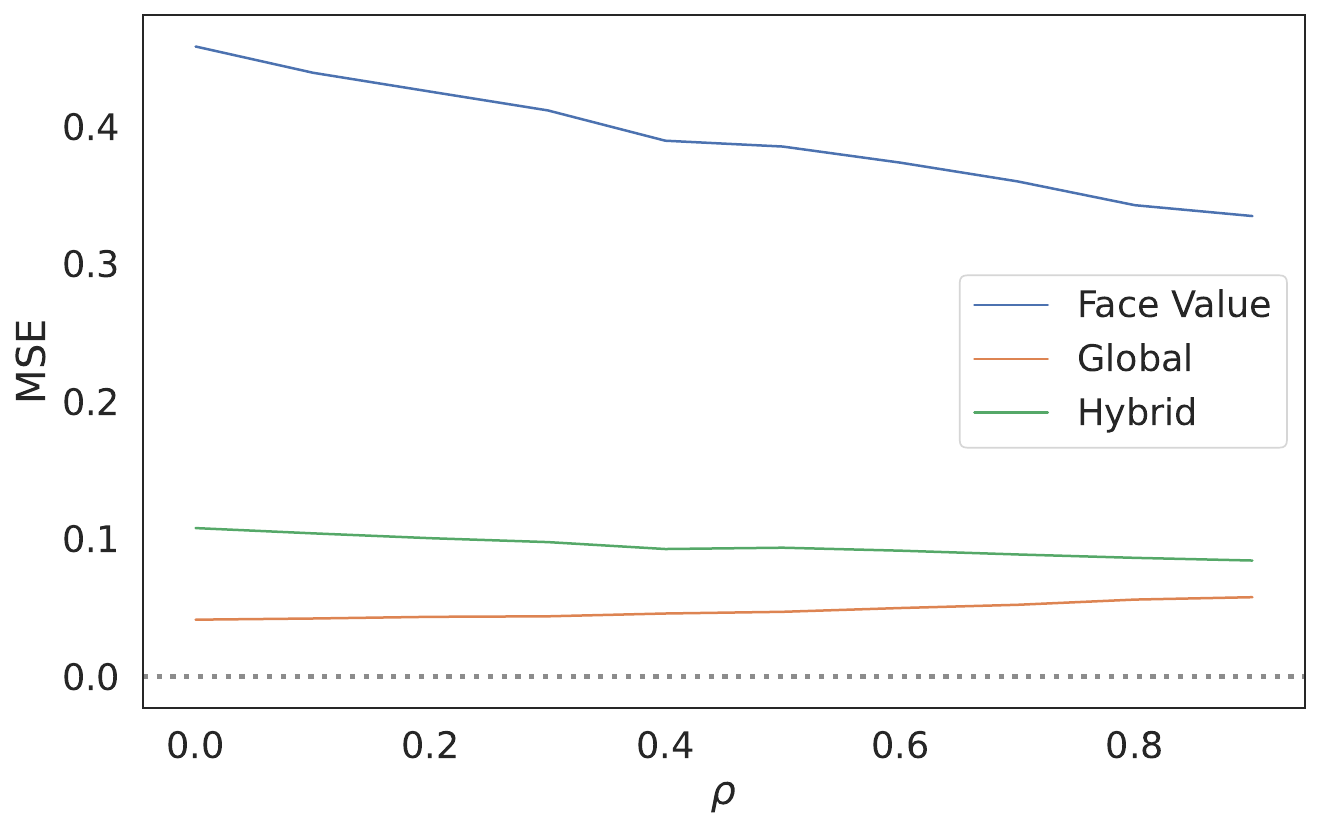}
        \label{subfig:mse_rho}
    }\hfill
    \subfloat{
        \includegraphics[width=0.3\textwidth]{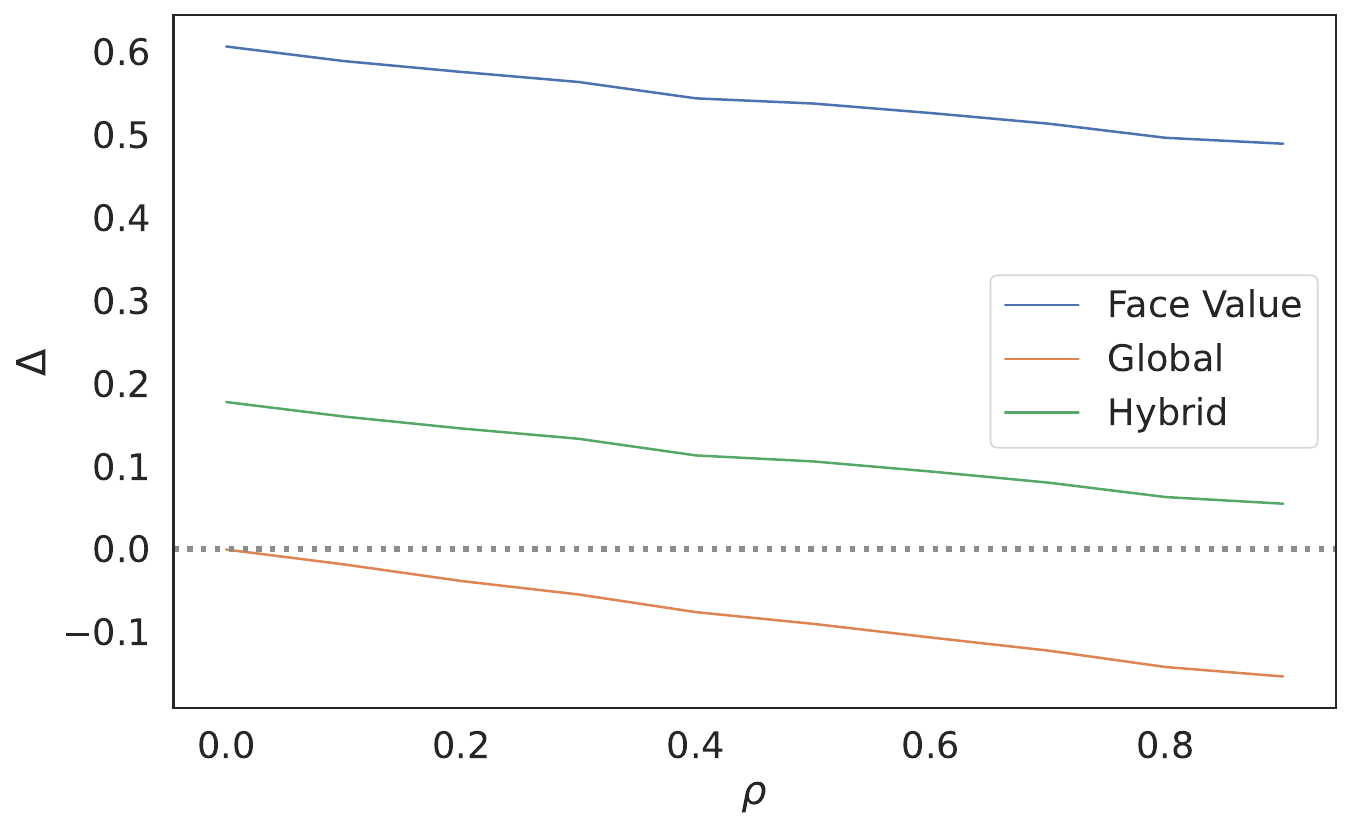}
        \label{subfig:bias_rho}
    }\hfill
    \subfloat{
        \includegraphics[width=0.3\textwidth]{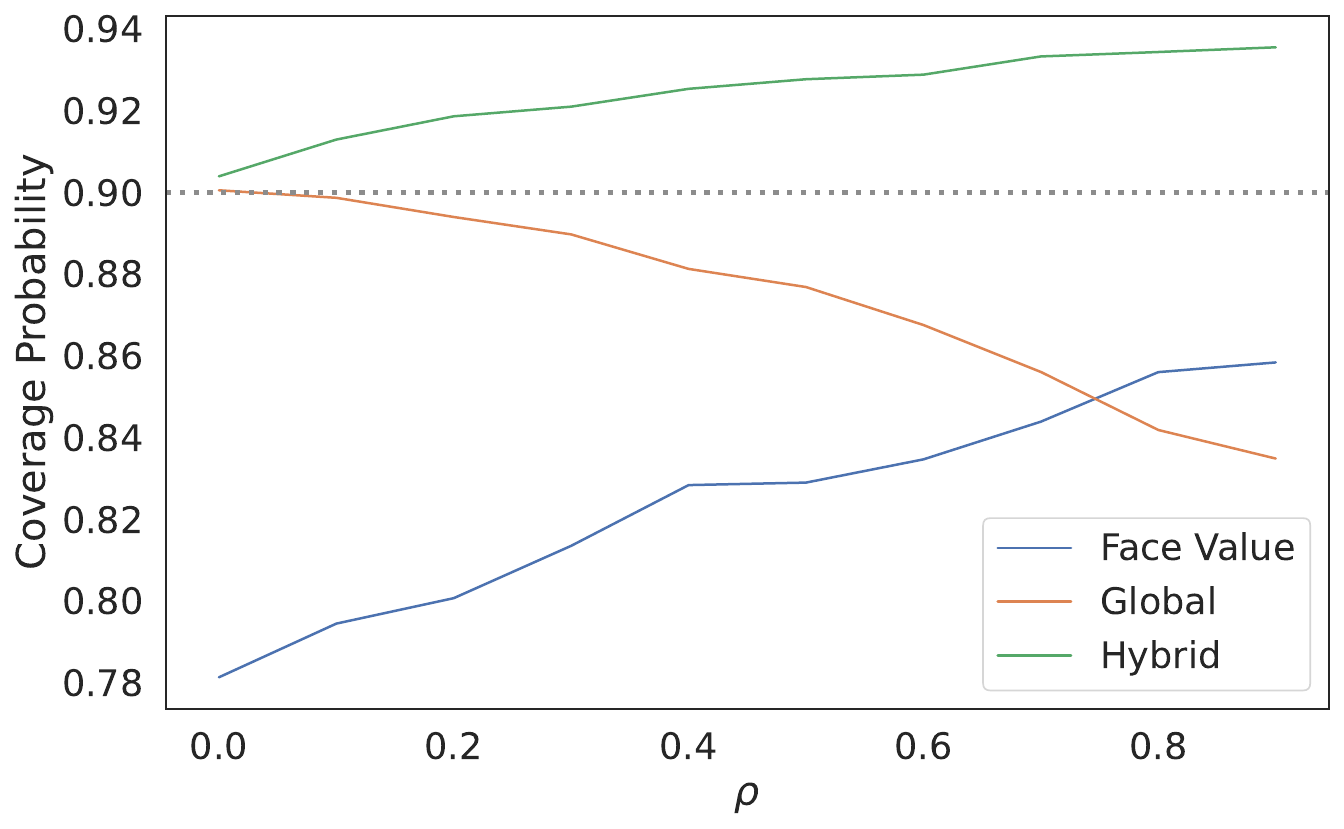}
        \label{subfig:coverage_rho}
    }
    \caption{MSE (left), bias ($\Delta$, center) and coverage probability (right) for face value, Bayesian with global shrinkage, and BHS approaches, as a function of prior mean ($\mu$, top), degrees of freedom ($\nu$, middle) and correlation ($\rho$, bottom).}
    \label{fig:misspecified}
\end{figure*}





\section{Empirical Analysis}

Motivated by the simulation performance of the BHS technique, particularly its robustness against model misspecification, we conduct an analysis in a real-world setting, using a series of online experiments from Meta Platforms' experimentation platform that serves multiple \textit{business units}.

We utilize a set of $N=167$ experiments across several business units, each comprising approximately $m_i=940$ experimental units, that are randomized between treatment and control\footnote{These experimental units can be clusters of units (e.g., users or advertisers).}. These experiments have \textbf{paired replication studies}: instances where the exact same treatment was re-tested on a distinct, subsequent sample. These replications provide a ``ground truth'' proxy that allows us to validate our estimates against realized future outcomes, applying the predictive checking strategies outlined in Section~\ref{section:validation}.

We compare the performance of the Face Value Frequentist estimator, the Bayesian global shrinkage estimator, and the BHS estimator. We use the Mean Absolute Error (MAE) relative to the replication result to evaluate performance \footnote{interval coverage is also evaluated across approaches which yields approximately well calibrated inference across the methods compared}.

\begin{table}[h!]
\centering
\caption{Summary of Performance Metrics -- MAE x 1000 -- across a collection of real-world experiments with paired replication studies.}
\vspace{-0.1cm}
\begin{tabular}{||c c||}
 \hline
Model & MAE \\ [0.5ex]
 \hline\hline
 Face Value & 1.280\\
 \hline
 Global Shrinkage & 1.121\\
 \hline
 Hybrid Shrinkage & 1.012\\ [1ex]
 \hline
\end{tabular}
\label{tbl: summary}
\end{table}

\begin{figure}
    \centering
       \includegraphics[width=0.9\linewidth]{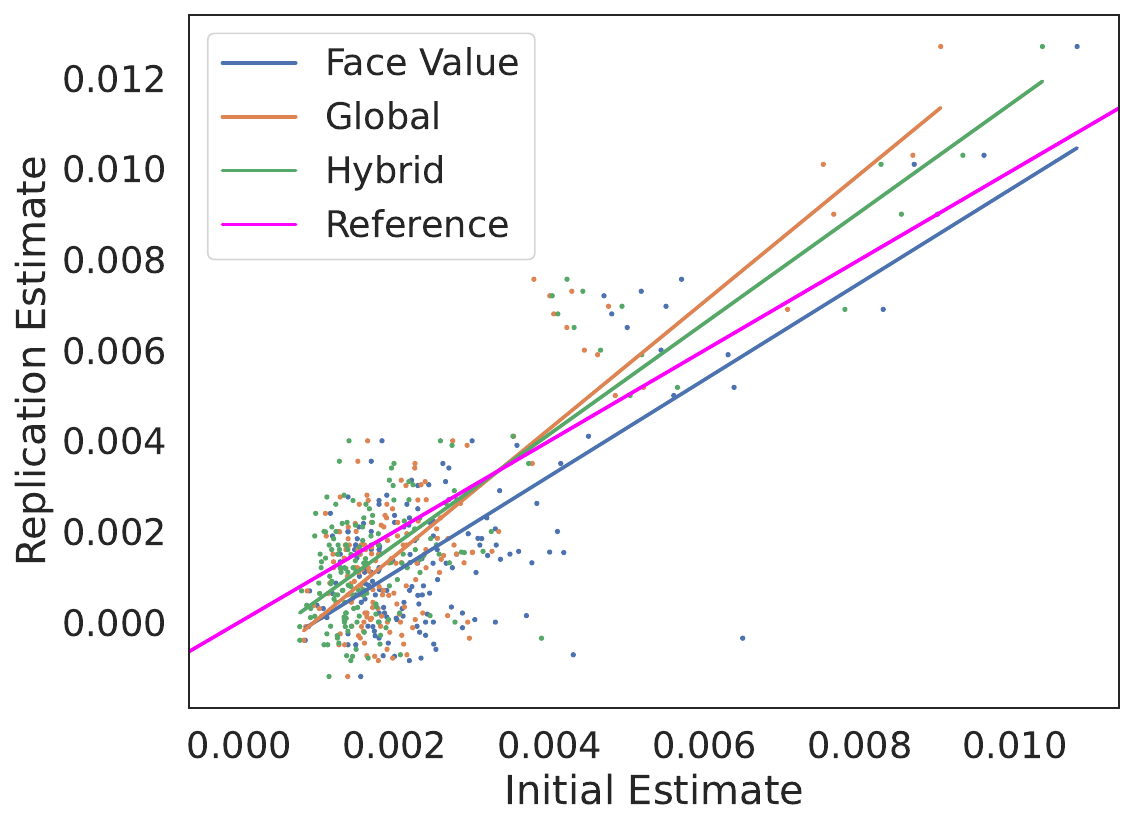}
    \caption{Comparison of estimators across a collection of real-world experiments with paired replication studies for two different business units.}
    \label{fig:comparison_proposed}
\end{figure}

\begin{figure}
    \centering
     \includegraphics[width=0.9\linewidth]{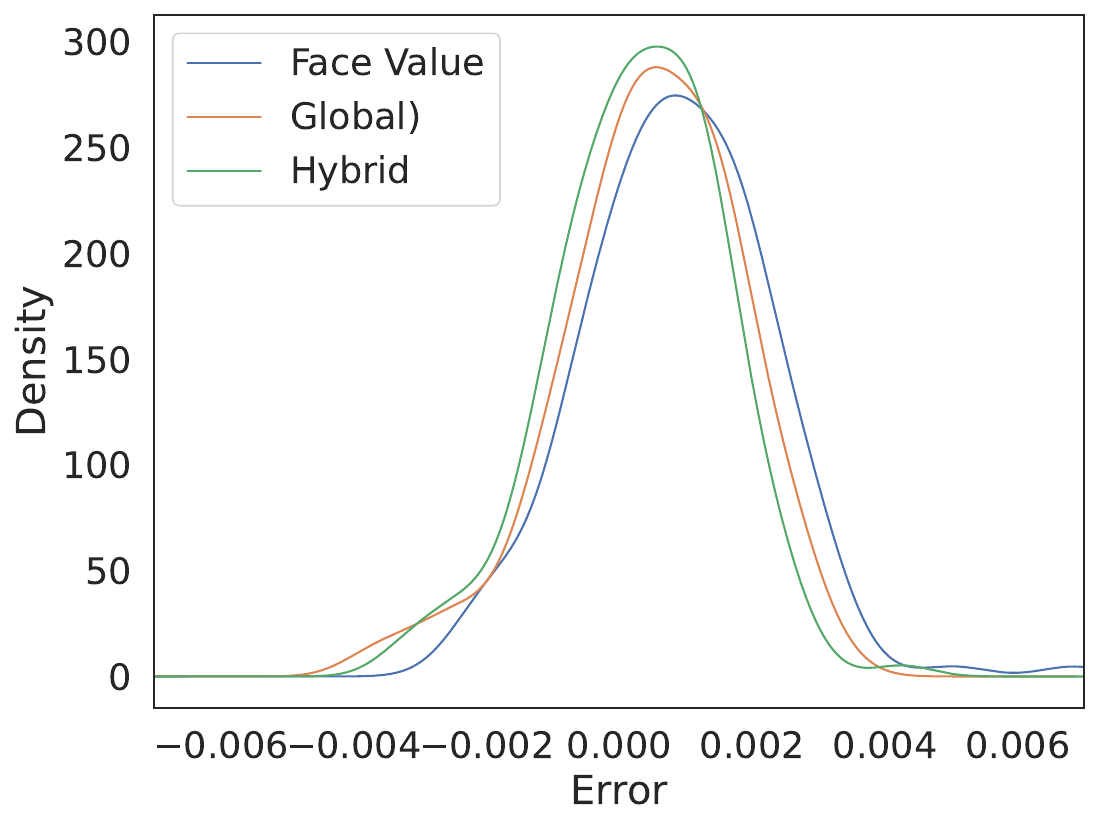}
    \caption{Distribution of errors for different estimators across a collection of real-world experiments with paired replication studies.}
    \label{fig:comparisn_errors}
\end{figure}

These results reinforce the findings of the simulation. As shown in Table~\ref{tbl: summary}, the BHS approach improves accuracy compared to both the Face Value estimator and the Bayesian approach with Global Shrinkage only. As shown in Figure~\ref{fig:comparison_proposed}, the Global Shrinkage estimator shrinks ``winners'' too strongly toward the global mean since it fails to account for experiment-level heterogeneity (i.e. heavy tails). In contrast, the BHS estimator accommodates these extremities. Hence, wherever important decision making is involved, a model that balances between learning across experiments while accommodating the individual behavior of specific experiments is recommended. An overview of the implementation for the experimentation platform at Meta Platforms is given in Appendix~\ref{ap:a1}.

\section{Conclusion}

A/B tests play an important role in data-driven decision making in many fields, including in the information technology domain as discussed here. However, as we have highlighted in this paper, experimental data is often used for both selection and inference, leading to biased estimates and invalid uncertainty quantification. This issue is increasingly relevant in application areas where systems are maturing and the relative effect of innovations is typically small.

To address this challenge, we proposed a Bayesian Hybrid Shrinkage (BHS) estimator and related uncertainty quantification technique that enables bias-corrected estimation with well-calibrated intervals under selection. Our approach remains robust against common types of misspecification—including heavy-tailed distributions and hidden selection—and outperforms key benchmarks in both simulation and empirical studies. By introducing local shrinkage factors, BHS bridges the gap between the computational efficiency of global methods and the flexibility of hierarchical models.

By providing a more accurate and reliable means of estimating treatment effects, our technique has the potential to improve decision making in a wide range of applications. Looking ahead, our proposed technique lays the foundation for future research on more complex experimental designs, such as sequential experiments, and more flexible means of parameterizing the analysis for identification and inference. We plan to extend our analysis in these directions, further advancing the state-of-the-art in A/B testing and data-driven decision making.

\newpage

\bibliographystyle{plainnat}
\bibliography{kdd_bib}

@article{stephens2017false,
  title={False discovery rates: a new deal},
  author={Stephens, Matthew},
  journal={Biostatistics},
  volume={18},
  number={2},
  pages={275--294},
  year={2017},
  publisher={Oxford University Press}
}

@article{rasines2022empirical,
  title={Empirical Bayes and Selective Inference},
  author={Rasines, Daniel Garc{\'\i}a and Young, G Alastair},
  journal={Journal of the Indian Institute of Science},
  volume={102},
  number={4},
  pages={1205--1217},
  year={2022},
  publisher={Springer}
}

@inproceedings{barajas2021online,
  title={Online advertising incrementality testing: practical lessons and emerging challenges},
  author={Barajas, Joel and Bhamidipati, Narayan and Shanahan, James G},
  booktitle={Proceedings of the 30th ACM International Conference on Information \& Knowledge Management},
  pages={4838--4841},
  year={2021}
}

@inproceedings{lee2018winner,
  title={Winner's curse: Bias estimation for total effects of features in online controlled experiments},
  author={Lee, Minyong R and Shen, Milan},
  booktitle={Proceedings of the 24th ACM SIGKDD international conference on knowledge discovery \& data mining},
  pages={491--499},
  year={2018}
}

@article{kessler2024overcoming,
  title={Overcoming the winner’s curse: Leveraging Bayesian inference to improve estimates of the impact of features launched via A/B tests},
  author={Kessler, Ryan},
  year={2024}
}

@inproceedings{ejdemyr2024estimating,
  title={Estimating the Returns from an Experimentation Program},
  author={Ejdemyr, Simon and Tingley, Martin and Shang, Yian and Brooks, Travis},
  booktitle={ACIC Conference},
  year={2024}
}

@book{kohavi2020trustworthy,
  title={Trustworthy online controlled experiments: A practical guide to a/b testing},
  author={Kohavi, Ron and Tang, Diane and Xu, Ya},
  year={2020},
  publisher={Cambridge University Press}
}

@article{gelman2014beyond,
  title={Beyond power calculations: Assessing type S (sign) and type M (magnitude) errors},
  author={Gelman, Andrew and Carlin, John},
  journal={Perspectives on psychological science},
  volume={9},
  number={6},
  pages={641--651},
  year={2014},
  publisher={Sage Publications Sage CA: Los Angeles, CA}
}

@article{gupta2019top,
  title={Top challenges from the first practical online controlled experiments summit},
  author={Gupta, Somit and Kohavi, Ronny and Tang, Diane and Xu, Ya and Andersen, Reid and Bakshy, Eytan and Cardin, Niall and Chandran, Sumita and Chen, Nanyu and Coey, Dominic and others},
  journal={ACM SIGKDD Explorations Newsletter},
  volume={21},
  number={1},
  pages={20--35},
  year={2019},
  publisher={ACM New York, NY, USA}
}

@article{andrews2024inference,
  title={Inference on winners},
  author={Andrews, Isaiah and Kitagawa, Toru and McCloskey, Adam},
  journal={The Quarterly Journal of Economics},
  volume={139},
  number={1},
  pages={305--358},
  year={2024},
  publisher={Oxford University Press}
}

@article{dawid1994selection,
  title={Selection paradoxes of Bayesian inference},
  author={Dawid, AP},
  journal={Lecture Notes-Monograph Series},
  pages={211--220},
  year={1994},
  publisher={JSTOR}
}

@article{yekutieli2012adjusted,
  title={Adjusted Bayesian inference for selected parameters},
  author={Yekutieli, Daniel},
  journal={Journal of the Royal Statistical Society Series B: Statistical Methodology},
  volume={74},
  number={3},
  pages={515--541},
  year={2012},
  publisher={Oxford University Press}
}

@article{woody2022optimal,
  title={Optimal post-selection inference for sparse signals: a nonparametric empirical Bayes approach},
  author={Woody, Spencer and Padilla, Oscar Hernan Madrid and Scott, James G},
  journal={Biometrika},
  volume={109},
  number={1},
  pages={1--16},
  year={2022},
  publisher={Oxford University Press}
}

@article{ding2018causal,
  title={Causal inference},
  author={Ding, Peng and Li, Fan},
  journal={Statistical Science},
  volume={33},
  number={2},
  pages={214--237},
  year={2018},
  publisher={JSTOR}
}

@article{van2021significance,
  title={The significance filter, the winner's curse and the need to shrink},
  author={van Zwet, Erik W and Cator, Eric A},
  journal={Statistica Neerlandica},
  volume={75},
  number={4},
  pages={437--452},
  year={2021},
  publisher={Wiley Online Library}
}

@article{rubin1998more,
  title={More powerful randomization-based p-values in double-blind trials with non-compliance},
  author={Rubin, Donald B},
  journal={Statistics in medicine},
  volume={17},
  number={3},
  pages={371--385},
  year={1998},
  publisher={Wiley Online Library}
}

@article{gelman1996posterior,
  title={Posterior predictive assessment of model fitness via realized discrepancies},
  author={Gelman, Andrew and Meng, Xiao-Li and Stern, Hal},
  journal={Statistica sinica},
  pages={733--760},
  year={1996},
  publisher={JSTOR}
}

@article{geyer1992practical,
  title={Practical markov chain monte carlo},
  author={Geyer, Charles J},
  journal={Statistical science},
  pages={473--483},
  year={1992},
  publisher={JSTOR}
}

@article{benjamini2020selective,
  title={Selective inference: The silent killer of replicability},
  author={Benjamini, Yoav},
  year={2020},
  publisher={PubPub}
}

@article{howes2024understanding,
  title={Understanding variability: the role of meta-analysis of variance},
  author={Howes, Oliver D and Chapman, George E},
  journal={Psychological Medicine},
  volume={54},
  number={12},
  pages={3233--3236},
  year={2024},
  publisher={Cambridge University Press}
}

@book{efron2012large,
  title={Large-scale inference: empirical Bayes methods for estimation, testing, and prediction},
  author={Efron, Bradley},
  volume={1},
  year={2012},
  publisher={Cambridge University Press}
}

@article{kleijn2012bernstein,
  title={The Bernstein-von-Mises theorem under misspecification},
  author={Kleijn, Bas JK and Van der Vaart, Aad W},
  year={2012}
}

@article{xu2011bayesian,
  title={Bayesian methods to overcome the winner's curse in genetic studies},
  author={Xu, Lizhen and Craiu, Radu V and Sun, Lei},
  journal={The Annals of Applied Statistics},
  pages={201--231},
  year={2011},
  publisher={JSTOR}
}

@incollection{rasines2022bayesian,
  title={Bayesian selective inference},
  author={Rasines, Daniel Garc{\'\i}a and Young, G Alastair},
  booktitle={Handbook of Statistics},
  volume={47},
  pages={43--65},
  year={2022},
  publisher={Elsevier}
}

@article{mandel2007statistical,
  title={On statistical inference under selection bias},
  author={Mandel, Micha and Rinott, Yosef},
  year={2007},
  publisher={Citeseer}
}

@article{zollner2007overcoming,
  title={Overcoming the winner’s curse: estimating penetrance parameters from case-control data},
  author={Z{\"o}llner, Sebastian and Pritchard, Jonathan K},
  journal={The American Journal of Human Genetics},
  volume={80},
  number={4},
  pages={605--615},
  year={2007},
  publisher={Elsevier}
}

@article{lind1991winner,
  title={The winner's curse: experiments with buyers and with sellers},
  author={Lind, Barry and Plott, Charles R},
  journal={The American economic review},
  volume={81},
  number={1},
  pages={335--346},
  year={1991},
  publisher={JSTOR}
}

@book{thaler1992winner,
  title={The winner's curse: Paradoxes and Anomalies of Economic Life},
  author={Thaler, Richard H},
  publisher={Free Press},
  year={1992},
  }

@article{ferguson2013empirical,
  title={Empirical Bayes correction for the winner's curse in genetic association studies},
  author={Ferguson, John P and Cho, Judy H and Yang, Can and Zhao, Hongyu},
  journal={Genetic epidemiology},
  volume={37},
  number={1},
  pages={60--68},
  year={2013},
  publisher={Wiley Online Library}
}

@article{leiner2023data,
  title={Data fission: splitting a single data point},
  author={Leiner, James and Duan, Boyan and Wasserman, Larry and Ramdas, Aaditya},
  journal={Journal of the American Statistical Association},
  pages={1--12},
  year={2023},
  publisher={Taylor \& Francis}
}

@inproceedings{Fiez2024Anduril,
 author = {Tanner Fiez and Houssam Nassif and Yu-Cheng Chen and Sergio Gamez and Lalit Jain},
 title = {Best of Three Worlds: Adaptive Experimentation for Digital Marketing in Practice},
 booktitle = {The Web Conference (WWW)},
 year = {2024},
 pages = {3586--3597},
}

@inproceedings{xu2013MABbias,
 author = {Xu, Min and Qin, Tao and Liu, Tie-Yan},
 booktitle = {Advances in Neural Information Processing Systems},
 title = {Estimation Bias in Multi-Armed Bandit Algorithms for Search Advertising},
 pages = {2400 -- 2408},
 year = {2013}
}

@article{nabi2022EB,
 author={Sareh Nabi and Houssam Nassif and Joseph Hong and Hamed Mamani and Guido Imbens},
 title={Bayesian Meta-Prior Learning Using {Empirical Bayes}},
 journal={Management Science},
 year = {2022},
 volume = {68},
 number = {3},
 pages = {1737--1755},
}

@inproceedings{Radwan2024Eval,
 author = {Mohamed A Radwan and Quinn Lanners and Jiasheng Zhang and Serkan Karakulak and Houssam Nassif and Murat Ali Bayir},
 title = {Counterfactual Evaluation of Ads Ranking Models through Domain Adaptation},
 booktitle = {Workshops of Conference on Recommender Systems (RecSys)},
 year = {2024},
}

@article{Li2027optimalPolicies,
 author={Li, Zhaoqi and Nassif, Houssam and Luedtke, Alex},
 title={Estimation of subsidiary performance metrics under optimal policies},
 journal={Statistica Sinica},
 year = {2027},
 volume={37},
 number={3},
}

\appendix
\section{Recommendation for System Design}
\label{ap:a1}
The Bayesian Hybrid Shrinkage (BHS) layer can be implemented as a non-invasive augmentation that can be applied on top of existing A/B testing infrastructure. Historical experimental data are aggregated offline to estimate the global hyperparameter $\tau$ along with $a, b$ and $m_0$ via Empirical Bayes methods. The hyperparameters are subsequently propagated to a distributed, low-latency configuration repository to be accessed as needed.

Upon request invocation from the experimentation user interface, a computationally efficient inference module accesses the pre-computed global hyperparameters and the experiment-level sufficient statistics, using them to derive posterior summaries under BHS. The sufficient statistics foundational definitions and aggregation procedures remain invariant. The user interface simply surfaces the BHS posterior means in lieu of the Face Value estimator.

This architectural decoupling of metric computation from statistical inference ensures backwards compatibility with legacy systems, while simultaneously facilitating the application of shrinkage estimators across the experimental portfolio. An example implementation schematic for the experimentation platform at [Redacted] is given in Figure \ref{fig:sch}.

\begin{figure}[hb]
    \centering
    \includegraphics[width=0.90\linewidth]{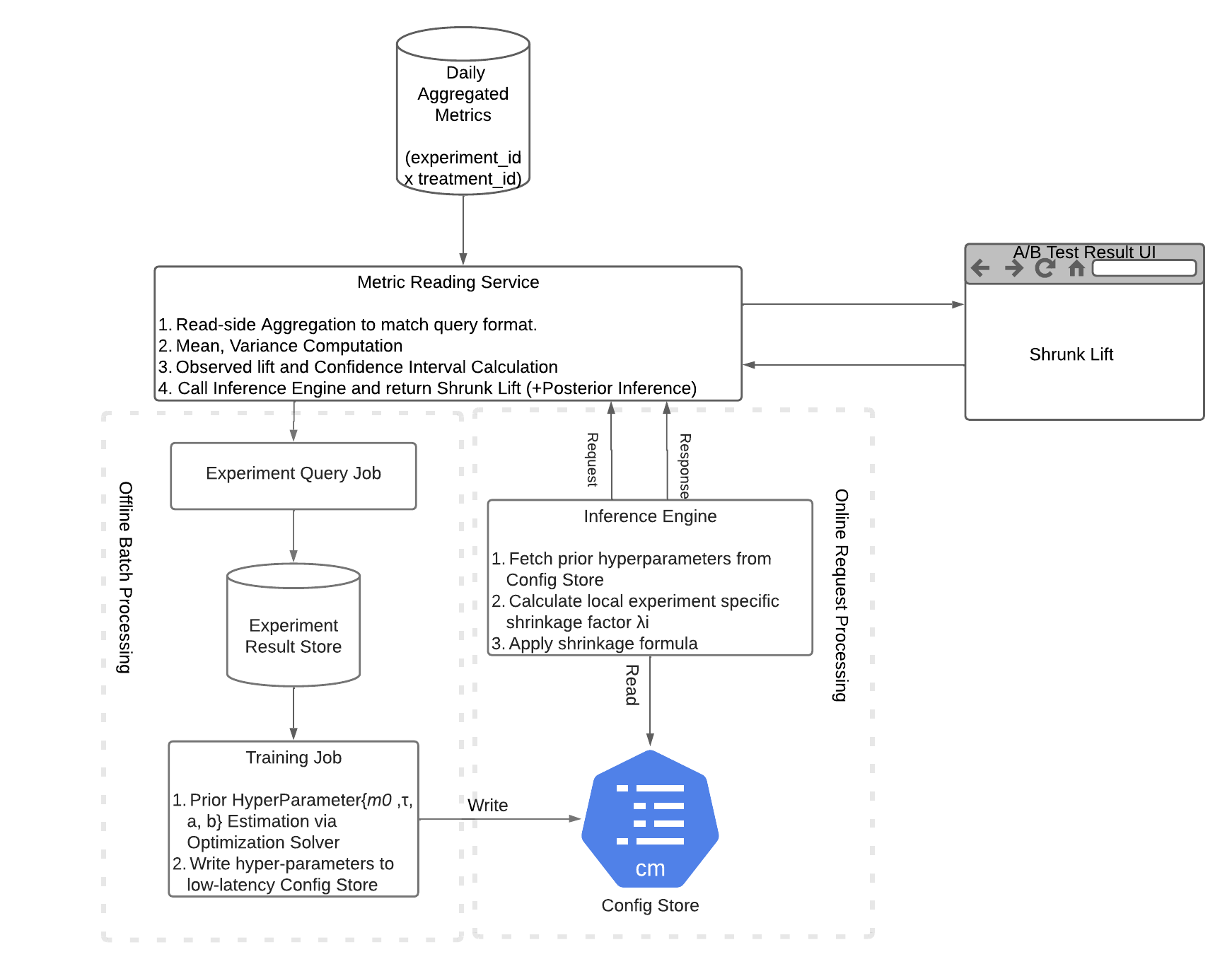}
    \caption{Schematic of implementation for a post-hoc BHS adjustment.}\label{fig:sch}
\end{figure}

\end{document}